\documentclass[a4paper,onecolumn,11pt,accepted=2022-01-22]{quantumarticle}
\pdfoutput=1
\usepackage[numbers]{natbib}
\usepackage{tikz}
\usepackage{amsmath}
\usepackage{amssymb}  
\usepackage{amsthm}
\usepackage{amsfonts}
\usepackage{graphicx}
\usepackage{bbm}
\usepackage{url}

\usepackage{braket}
\usepackage{ upgreek }
\usepackage{dsfont}
\usepackage{makecell}
\usepackage{authblk}
\usepackage[plain]{fancyref} %
\usepackage{algorithm}
\usepackage{algpseudocode}
\usepackage{comment}

\usepackage{float}
\newfloat{algorithm}{t}{lop}

\usepackage{hyperref}

\theoremstyle{plain}               
\newtheorem{thm}{Theorem}[section]

\newtheorem{lem}{Lemma}[section]

\newtheorem{prop}{Proposition}[section]

\newtheorem{defi}{Definition}[section]
\newtheorem{example}{Example}[section]
\theoremstyle{remark}

\newcommand*{\fancyrefthmlabelprefix}{thm}
\newcommand*{\fancyreflemlabelprefix}{lem}
\newcommand*{\fancyrefcorlabelprefix}{cor}
\newcommand*{\fancyrefdefilabelprefix}{defi}
\frefformat{plain}{\fancyreflemlabelprefix}{lemma\fancyrefdefaultspacing#1}
\Frefformat{plain}{\fancyreflemlabelprefix}{Lemma\fancyrefdefaultspacing#1}
\frefformat{plain}{\fancyrefthmlabelprefix}{theorem\fancyrefdefaultspacing#1}
\Frefformat{plain}{\fancyrefthmlabelprefix}{Theorem\fancyrefdefaultspacing#1}
\frefformat{plain}{\fancyrefcorlabelprefix}{corollary\fancyrefdefaultspacing#1}
\Frefformat{plain}{\fancyrefcorlabelprefix}{Corollary\fancyrefdefaultspacing#1}
\frefformat{plain}{\fancyrefdefilabelprefix}{definition\fancyrefdefaultspacing#1}
\Frefformat{plain}{\fancyrefdefilabelprefix}{Definition\fancyrefdefaultspacing#1}
\newcommand*{\fancyrefalglabelprefix}{alg}
\newcommand*{\frefalgname}{algorithm}
\newcommand*{\Frefalgname}{Algorithm}
\frefformat{plain}{\fancyrefalglabelprefix}{%
  \frefalgname\fancyrefdefaultspacing#1%
}%
\Frefformat{plain}{\fancyrefalglabelprefix}{%
  \Frefalgname\fancyrefdefaultspacing#1%
}%

\newcommand*{\fancyrefapplabelprefix}{app}
\newcommand*{\frefappname}{appendix}
\newcommand*{\Frefappname}{Appendix}
\frefformat{plain}{\fancyrefapplabelprefix}{%
  \frefappname\fancyrefdefaultspacing#1%
}%
\Frefformat{plain}{\fancyrefapplabelprefix}{%
  \Frefappname\fancyrefdefaultspacing#1%
}%

\definecolor{Green}{HTML}{00AD69}  %

\def\beq{\begin{equation}}
\def\eeq{\end{equation}}
\def\bq{\begin{quote}}
\def\eq{\end{quote}}
\def\ben{\begin{enumerate}}
\def\een{\end{enumerate}}
\def\bit{\begin{itemize}}
\def\eit{\end{itemize}}

\def\lb{\left(}
\def\rb{\right)}

\def\l|{\left|}
\def\r|{\right|}

\newcommand\Z{\mathbbm{Z}}
\newcommand\R{\mathbbm{R}}
\newcommand\N{\mathbbm{N}}
\newcommand\M{\mathcal{M}}

\newcommand{\ketbra}[1]{|#1\rangle\langle#1|}
\newcommand{\tr}[1]{\text{tr}\lb#1\rb}
\newcommand{\one}{I}

\newcommand{\rl}[2]{S\lb#1\|#2\rb}

\newcommand{\cO}{\mathcal{O}}

\newcommand{\cZ}{\mathcal{Z}}
\newcommand{\bE}{\mathbb{E}}
\newcommand{\bP}{\mathbb{P}}
\newcommand{\cU}{\mathcal{U}}

\begin{document}

\title{A game of quantum advantage: \newline linking verification and simulation}

\author[1,2]{Daniel Stilck Fran\c{c}a}
\author[3]{Raul Garcia-Patron}

\affil[1]{\small QMATH, Department of Mathematical Sciences, University of Copenhagen, Denmark}
\affil[2]{\small Univ Lyon, ENS Lyon, UCBL, CNRS, Inria, LIP, F-69342, Lyon Cedex 07, France}

\affil[3]{\small School of Informatics, University of Edinburgh, Edinburgh EH8 9AB, UK}

\date{}

\maketitle

\abstract{We present a formalism that captures the process of proving quantum superiority
to skeptics as an interactive game between two agents, supervised by a referee. 
The model captures most of the currently existing quantum advantage verification techniques.
In this formalism, Bob samples from a distribution on a quantum device that is supposed to demonstrate a quantum advantage. 
The other player, the skeptical Alice, is then allowed to propose mock distributions supposed to reproduce Bob's device's statistics. 
Bob then needs to provide witness functions to prove that Alice's proposed mock distributions 
cannot properly approximate his device. Within this framework, we establish three results.
First, for random quantum circuits,
Bob being able to efficiently distinguish his distribution from Alice's implies efficient approximate simulation of the distribution. Secondly, 
finding a polynomial time function to distinguish the output of random circuits from the uniform distribution 
can also spoof the heavy output generation problem in polynomial time. This pinpoints that exponential resources
may be unavoidable for even the most basic verification tasks in the setting of random quantum circuits. 
Finally, by employing strong data processing inequalities, our framework allows us to analyse the effect of noise on classical simulability and verification of more general near-term quantum advantage proposals. }
\section{Introduction}

The transition from the reign of classical computers to quantum superiority is expected
not to be a singular event but rather a process of accumulation of evidence. It will most probably happen through an iterative process of claims of proofs and refutations until a consensus is reached among the scientific community. 
A few years back the series of claims of the advantage of quantum annealers followed by rebuttals and an intense debate inside the quantum computation community~\cite{shin2014quantum,Smolin_2014,Boixo_2014,wang2013comment} 
can be seen as an example of that.  Similarly, recent claims of quantum advantage \cite{Arute2019,wu2021strong,Zhong2020,zhu2021quantum} were followed by a growing interest in its potential simulation by a classical device~\cite{huang2020classical,pednault2019leveraging,2109.11525}.

Ideally, one would like to demonstrate the advantage  of quantum computers solving 
a well-established hard computational problem, such as factoring large numbers or simulating
large-sized molecules. Such demonstration will most likely need a fault-tolerant quantum computer, 
which will not be available in the near future. Thus, a lot of attention has been focused in the last years
on quantum advantage proposals based on sampling from the outcomes of random quantum circuits, a 
task considered achievable. This effort culminated in landmark experiment of~\cite{Arute2019}.
and its more recent followups \cite{wu2021strong,Zhong2020,zhu2021quantum}.

The  classical hardness of computing the outcome probability of random circuits has been reduced to standard complexity-theoretic assumptions in various settings~\cite{harrow2017quantum,bouland2018quantum,lund2017quantum,movassagh2018efficient,bremner2011classical,aaronson2016complexity,bremner2016achieving,haferkamp2019closing}. 
For instance, in~\cite{bouland2018quantum} the authors prove that it is $\sharp P$ hard to compute the exact probability of the outputs of random quantum circuits for a fraction of $\frac{3}{4}+\textrm{poly}(n)^{-1}$ of instances of random quantum circuits on $n$ qubits. This result can be extended to devices with very low noise by assuming a couple of widely-accepted
conjectures. Despite this significant progress in putting the classical hardness of sampling from a distribution close to the outputs of the  random circuit on solid grounds, equivalent hardness statements are not known to hold for the levels of noise that affect current quantum computing architectures. Moreover, certifying closeness to the ideal distribution in total variation distance requires an exponential number of samples~\cite{PhysRevLett.122.210502}. Thus, it is both not feasible to verify closeness in total variation, and the actual distance is unlikely to be in the regime of current quantum advantage experiments.

These shortcomings have shifted the interest to benchmarks of advantage that are thought to be more robust against noise and that are known to be verifiable with a feasible number of samples, albeit not computationally efficient. Prominent examples are the heavy output generation problem (XHOG)~\cite{aaronson2016complexity,aaronson2019classical} and the related  linear cross-entropy benchmarking (linear XEB) fidelity~\cite{Neill_2018,Boixo_2018,Arute2019}.
The recent quantum advantage experiments used linear XEB as a benchmark for the quantum
state generated by 50 to 60 qubit devices. However, these approaches have two main drawbacks. First, they require the computation of the probability of sampled strings under the circuit's ideal distribution,
which consumes a running time growing exponentially with the system's size. Secondly, the number of required samples grows exponentially with the size of the system for a constant gate error probability~\cite{Arute2019}. Thus, the linear XEB verification approach requires us to be in the ``sweet spot" where both the number of samples needed given the noise level and the size of the circuit are not too large to render the verification impossible. This approach is not scalable to larger system sizes with current levels of noise. Besides, it is still unknown how to reduce the hardness of the heavy output generation problem (XHOG)~\cite{aaronson2016complexity,aaronson2019classical,2008.08721} to standard complexity-theoretic assumptions.

This difficulty of finding efficient certification protocols and benchmarks for random circuit sampling 
extends to other proposals of near-term quantum advantage, such as boson sampling~\cite{Aaronson_2014}. This has sparked interest in simpler sanity checks, such as efficiently distinguishing the output distribution from the uniform and other  ``easy" distributions~\cite{boson_far_uniform,Carolan_2014,phillips_benchmarking_2019,Spagnolo_2014}. 
As mentioned in~\cite{PhysRevLett.122.210502}, these verification forms are manifestly weaker than certifying the total variation distance and do not preclude the possibility of the device being classically simulable. However, in the setting of random circuit sampling, not even an efficient verification test that allows us to distinguish the outputs from the uniform distribution is known. Furthermore, although efficient verification protocols for quantum computation exist \cite{mahadev_classical_2018,chung_constant-round_2020}, they are likely to require fault-tolerance and are beyond what can be achieved with near-term implementations.

Independent of the recent quantum advantage experiment, the development of more efficient certification techniques of quantum advantage that can be scaled with the increasing size of the quantum computer is 
an area of relevance for near-term quantum computation~\cite{2005.04826}. In parallel, it is important to develop a better understanding of how noise reduces the power of quantum computers and how noise affects quantum advantage proposals or near-term applications of quantum devices.

To these ends, in this work, we envision this certification process as an interactive game between two agents, 
Alice that uses classical computing resources, and Bob that holds a (noisy) quantum computer and wants to convince Alice of its computational advantage. They are both supervised by the referee Robert. To win, Bob has to find functions that allow him to efficiently distinguish the output of his device from every alternative distribution Alice proposes. In turn, Alice needs to propose alternative distributions that approximate the statistics obtained by Bob's quantum computer; otherwise, she loses the game.

Central to our result is the connection between the mirror descent algorithm~\cite{tsuda_matrix_2005,Bubeck2015} and the proposed framework of the certification game. This allows us to connect distinguishing probability distributions from a target quantum distribution and learning an approximate classical description of the latter. Furthermore, as we will see, mirror descent is particularly well-suited to learning distributions of high entropy, which is the case for NISQ devices and current quantum advantage proposals.
Our framework is inspired by a recent result of the authors~\cite{2009.05532}. There, we show how to use mirror descent, strong data processing inequalities and related concepts to analyse the performance of noisy quantum devices performing optimization. 
In contrast, this article's main result also holds for noiseless circuits and our overarching goal is to formally link the hardness of verification of quantum advantage proposals and their approximate classical simulation.

\section{Summary of results}\label{sec:summary}

We envision a quantum advantage demonstration as a game played between Bob, who is sampling from a (noisy) quantum circuit and claiming that it demonstrates a quantum advantage, and Alice, who is skeptical of Bob's claim and defends that her classical computer can mimic Bob's behaviour efficiently.

As Bob is the person looking to convince the others that his quantum computer has an advantage over Alice's classical device, we place the burden of the demonstration on him. In addition, Bob publicly discloses a description of the hardware and the quantum algorithm his device is implementing. 
The game consists of different rounds at which Alice can propose an alternative hypothesis to the claim that Bob has achieved a quantum advantage. At the beginning of the game, they both agree on a distinguishability parameter $\epsilon>0$, which captures
how close Alice needs to match Bob's result, and confidence probability $\delta$,
which captures the probability of the outcome of the game being correct. 
In what follows, we denote the probability distribution Bob is sampling from by $\nu$. 

\subsection{The quantum advantage game}

In what follows we present a framework that captures most of the existing
quantum advantage verification protocols to date as an interactive game between players.

At the beginning of the first round, Alice discloses an alternative hypothesis to quantum advantage in the form of a randomized classical algorithm sampling from a given distribution $\mu_0(x)$, for example the uniform distribution.
It is then Bob's role to refute Alice's proposal and show that his distribution is at least $\epsilon$ away from Alice's in total variation distance. As we will discuss in more detail later in Section~\ref{sec:preliminaries}, certifying this distance constraint is equivalent to Bob providing a function $f_1:\{0,1\}^n\to[-1,1]$
such that
\begin{align}\label{equ:candistinguish}
\left|\mathbb{E}_{\mu_0}(f_{1})-\mathbb{E}_{\nu}(f_{1})\right|\geq\epsilon.
\end{align}
Alice is then allowed to update her hypothesis to $\mu_1(x)$ given the information gained from
the first round of the game and sample from a distribution $\mu_1$ that could potentially satisfy
\begin{align*}
\left|\mathbb{E}_{\mu_1}(f_{1})-\mathbb{E}_{\nu}(f_{1})\right|\leq\epsilon.
\end{align*}
If so, Bob then needs to refute this mock distribution, providing a new witness $f_2$.
Alice's new distribution then needs to pass both tests for the functions $f_1$ and $f_2$. The game continues with Bob proposing new distinguishing functions $f_{t+1}$ and Alice mock distributions $\mu_t$ that approximate all previous expectation values up to $\epsilon$.

The game ends if one of two players is declared defeated following a set of previously established rules.
For example, Alice could concede defeat if she takes too much time to propose a new candidate or sample from it. Similarly, Bob could be forced to acknowledge his defeat if 
he takes too long to offer a new witness that challenges Alice. In some sense, the rules should be consistent with the
process of building community consensus on the validity of a quantum advantage result.

It is important to remark that the condition in Eq.~\eqref{equ:candistinguish} must be checked by a referee Robert, to whom Alice and Bob provide samples at each round.  That is, they give Robert enough samples of their distributions to estimate the expectation values by computing the empirical average for the functions $f_i$ on the samples. The required number of samples required to be confident that the condition is satisfied up to a small failure probability can be estimated by e.g. an application of H\"offding's inequality, as we show later. 
Note that this guarantees that Alice cannot cheat by using Bob's samples to find better distributions. For instance, if Bob's quantum computer is solving an NP problem, then knowing the samples themselves would give Alice an efficient classical strategy. 
In order to suppress statistical anomalies, we also make the number of samples depend on the current round. More specifically, the number of samples for round $t$ should be $\cO(t\epsilon^{-2}\log(t\delta^{-1}))$. This choice ensures that the overall probability of an error occurring remains $\cO(\delta)$. 

In order for the verification game to be scalable, we may further request that sampling from Alice's distributions, evaluating Bob's test functions $f_i$ and the size of the messages sent to Robert to be tasks that need to be achieved efficiently with the resources at hand. In what follows, we define efficient 
as a consumption of resources that scales polynomially with the size of the problem, i.e., the number of qubits of the quantum computer. But a more at hands definition, where we impose constraints on their size and time of computation justified by the state of the art of classical computing hardware, is also compatible with this framework and most probably be the definition used in any real experimental demonstration. Test functions that are not scalable, such as the XEB used in random quantum circuit experiments,
are also contained in our game framework after some modifications. However, our no-go results do not directly apply to them as they focus on efficient functions. Indeed, the non-scalability of benchmarks like the XEB makes it impractical for future quantum computers of larger sizes than today's.

The framework presented above is quite general, capturing most of the existing
quantum advantage verification proposals to date. 
It is natural to ask how to phrase some current quantum advantage tests and attempts to spoof them within our framework. 
In the examples below we illustrate how to use the formalism to describe the scenario in which we benchmark against supposedly better solutions to NP-complete optimization problems, sampling from random quantum circuits and briefly discuss the case of boson sampling machines.

\subsection{Verification of NP problems}

As an example, let us consider the optimization version of an NP-complete problem, such as MAXCUT on a $\Delta$-regular graph with $\Delta\geq3$:
\begin{example}[MAXCUT]\label{ex:maxcut}
suppose that Bob claims his quantum computer can achieve a better value for an NP optimization problem, say MAXCUT, than Alice's classical computer. Recall that for a graph $G=(V,E)$ with $n$ vertices and maximum degree $\Delta$, MAXCUT of the graph can be cast as finding the maximum over $\{0,1\}^n$ of the function
\begin{align}\label{equ:functionmaxcut}
f_G(x)=\frac{1}{n\Delta}\sum\limits_{(i,j)\in E} (1-\delta(x_i,x_j)).
\end{align}
Here, the normalization $n\Delta$ ensures that $0\leq f(x)\leq 1$. Thus, in this case, Bob can propose the function $f_G$ to distinguish his distribution from Alice's classical computer. If the average value for MAXCUT he can achieve is indeed at least $\epsilon$ better than what Alice can achieve, he wins the game.  On the other hand, if classical methods can yield a better cut than Bob's quantum computer, then $f_G$ cannot claim to have achieved a quantum advantage. Note that our choice of normalization in Eq.~\eqref{equ:functionmaxcut} implies that an additive error $\epsilon$ on approximating the expectation value of $f_G$ implies a multiplicative error of order $\epsilon$ for the cut's value.
\end{example}
As exemplified above, for NP optimization problems, there is a clear choice for which function Bob should propose and it can be computed in polynomial time.

At first sight, the possibility of the game always requiring an exponential number of rounds seems a plausible outcome. However, there is an update rule for Alice's distribution that will lead to the game having at most $\cO( n\epsilon^{-2})$ rounds, where $n$ is the number of qubits of Bob's device, as we explain below. 
Note that we do not claim that this update rule will always lead to Alice winning, only that it will %
define a finite series of probability distributions that converge to the one the quantum device.
The key question is whether Alice can sample from those explicit distributions efficiently or not.
Interestingly, below we will show that this leads to a successful strategy against random circuits
under the condition Bob provides the efficient functions $f_t$.

This update rule uses the connection between our certification game and mirror descent with the von Neumann entropy as potential~\cite{Bubeck2015}, a method to learn probability distributions by approximating them by a sequence of Gibbs distributions. In a nutshell, Alice can exploit each test function $f_t$ that Bob provides to improve her guess of the distribution $\nu$. The updates are of the form $\mu_{t+1}\propto e^{\log(\mu_t)-\tfrac{\epsilon}{4} f_t}$. She can use her method of choice to sample from the Gibbs distribution, such as rejection sampling. One can then show that at every round of the game, Alice's gets closer to the ideal distribution by at least a finite amount, converging to $\mu_t$ being $\epsilon$-close to the ideal quantum distribution $\nu$ in a finite number of rounds. In fact, regardless of the distribution Bob is sampling from, if Alice chooses to use mirror descent to update her distribution, then it follows from standard properties of mirror descent that the game will end after at most $8n\epsilon^{-2}$ rounds. 
We explain this in more detail in Section~\ref{sec:simulability} and refer to Appendix \ref{sec:mirror_descent_basics} for a discussion and proof of its basic properties. 
The caveat is that the knowledge from which Gibbs distribution Alice needs to sample does not guarantee that she can do it efficiently. 
Let us exemplify this again with MAXCUT:
\begin{example}[MAXCUT continued-mirror descent and simulated annealing]
In the same setting as in Example~\ref{ex:maxcut}, one can show that if Alice uses mirror descent to update her probability distributions, her sequence of probability distributions $\mu_t$ is given by:
\begin{align*}
\mu_t(x)=\frac{e^{\frac{t}{4\epsilon}f_G(x)}}{\mathcal{Z}_t},\quad \mathcal{Z}_t=\sum_{x\in\{0,1\}^n}e^{\frac{t}{4\epsilon}f_G(x)}.
\end{align*}
That is, her strategy will be akin to performing simulated annealing to try to obtain a better value of MAXCUT. This is one of the most widely used methods for combinatorial optimization~\cite{Kirkpatrick_1983}. If one picks $t$ large enough, $\mu_t$ is guaranteed to be sharply concentrated around the maximum of $f_G$, but the complexity of sampling from $\mu_t$ increases accordingly and at some point Alice won't be able to sample from it anymore. We refer to e.g.~\cite[Chapter 28]{grotschel2012geometric} and references therein for a discussion of the application of simulated annealing to combinatorial optimization. We can interpret the parameter $t/4\epsilon$ as the inverse temperature $\beta$.
Thus, if Alice picks the mirror descent strategy, she would win if a classical Monte Carlo algorithm has a performance comparable to that of Bob's quantum computer.
\end{example}

\subsection{Random Quantum Circuits}

Random quantum circuits use a more sophisticated benchmarking strategy based on the linear cross-entropy.

\begin{example}[linear cross-entropy, spoofing it and correlators]\label{exa:rdqcs}
The current approach to benchmark quantum advantage experiments based on sampling from random quantum circuits is the linear cross-entropy~\cite{Neill_2018,Arute2019}. Given the outcome distribution of the ideal circuit $\nu$ and another distribution $\mu$, its value is given by:
\begin{align}\label{equ:cross_entropy0}
    \mathcal{F}_{\operatorname{XEB}}(\mu)=2^n\mathbb{E}_{\mu}\lb \nu\rb-1.
\end{align}
We discuss this metric thoroughly in Sec.~\ref{sec:hoguniform}, but roughly speaking the goal of this benchmark is to sample from a distribution $\mu$ that achieves $\mathcal{F}_{\operatorname{XEB}}(\mu)>0$. Note that it corresponds to the expectation value of the function
\begin{align}\label{equ:function_xhog}
f(x)=2^n\nu(x)-1.
\end{align}
\end{example}

In principle, the function defined in Eq.~\eqref{equ:function_xhog} does not fit our framework, as it could in principle take values larger than $1$. However, as we explain in Appendix.~\ref{app:heavyoutputgeneration_function}, under some assumptions that are believed to hold for outputs of random quantum circuit it is possible to show that the suitably cut-off function
\begin{align}\label{equ:cutoffcross}
f_r(x)=r^{-1}(\min\{2^n\mathbb{E}_{\mu}\lb \nu\rb,r\}-1)
\end{align}
for $r=\cO(1)$ can be used to approximate the value Eq.~\eqref{equ:cross_entropy0}. Furthermore, in Prop.~\ref{equ:distinguish_uniform} we show that the function in Eq.~\eqref{equ:cutoffcross} can also be used to distinguish the output of the circuit from the uniform distribution. Thus, we see that this commonly used benchmark also fits our framework and Bob could propose the variation of the linear cross-entropy in Eq.~\eqref{equ:cutoffcross} during the first round, although it is not efficiently computable.

Recently, tensor network contraction techniques have been proposed to spoof this benchmark, which would correspond to Alice passing the first round of the game if Bob proposes the function in Eq.~\eqref{equ:cutoffcross}. Let us now show how our framework could be used to provide Bob with extra functions to win against the approach championed in~\cite{pan2021simulating}. Roughly speaking, in that paper the authors fix the outcome of $k$ out of the $n$ qubits to some fixed output, say $\ket{0}^{\otimes k}$, and then search for strings with higher than average probability under $\nu$ in the space of strings with those outcomes fixed. As fixing some outcomes significantly reduces the computational cost of computing outcome probabilities, the authors are then able to generate many samples which have an expectation value for the linear cross-entropy that is close to the value reported in \cite{Arute2019}. If the authors of~\cite{pan2021simulating} were playing Alice with the strategy outlined in that paper, then Bob could resort to a simple strategy to distinguish his distribution from Alice's: using correlators. Let $i$ be one of the $k$ qubits always set to $\ket{0}$. We then let the function for the second round be given by $f_2(x)=1$ if $x_i=0$ and $0$ else. If Bob is sampling from the output distribution $\nu$ of a random quantum circuits, it will be the case that $\mathbb{E}_{\nu}(f_2)\simeq \tfrac{1}{2}$ up to exponentially small corrections with high probability, whereas for Alice's distribution $\mu_0$ we have $\mathbb{E}_{\mu_0}(f_2)=1$. Thus, this way Bob can easily distinguish his distribution from Alice's and protect himself from spoofs based on strategies like that of~\cite{pan2021simulating}.

\subsection{Gaussian Boson Sampling}
In the context of Boson sampling, correlators have been championed as a benchmark of the quantum advantage certification of the devices~\cite{phillips_benchmarking_2019}. That is, one computes some $k$-point correlation function of the ideal outcome distribution and compares it to the output of the device. Such tests easily fit into our framework. To see this, suppose we wish to consider a $2$-point correlation function on the first two bits. In that case, we could just pick the function $f(x)=\delta_{x_1,x_2}$, which satisfies the conditions discussed before. Applying mirror descent to the case of the correlators then gives rise to a classical Gibbs state that reproduces the local correlators of the ideal distribution. Interestingly, this strategy was recently used in~\cite{2109.11525}, where the authors observed that already fitting to some few-body correlators seems sufficient to obtain a better approximation in total variation to the true distribution than the noisy quantum device of~\cite{Zhong2020}.

\subsection{Summary and discussion}

Thus, we see that our framework is able to recover many of the strategies currently used in the literature to benchmark quantum advantage proposals, besides also giving advice as to how to refute spoofing strategies. Moreover, the mirror descent approach can also give rise to competitive spoofing techniques, as observed in~\cite{2109.11525}.

However, it is important to notice that our framework does not cover the most general efficient procedure to distinguish probability distributions. Indeed, our framework only includes procedures that use the empirical averages of single samples to distinguish distributions. Such a setting is very close in spirit to that of statistical queries~\cite{Kearns_1998,stat_query}.
However, a more general efficient procedure could apply an efficient function that depends on a polynomial number of samples to try to distinguish the distributions. In 
Appendix~\ref{sec:gen_limitations} we discuss possible extensions and limitations of our results in this direction.

\subsection{Main results}
As we anticipated, this framework of quantum advantage certification allows us to prove three main results on the impossibility of Bob winning using efficiently computable test functions, a connection between the HOG conjecture and the indistingushability from uniform distribution and an analysis of the effects of hardware errors on a quantum advantage verification protocol.

\subsubsection{Bob can not win with efficient distinguishing functions}
Our first result is that for random circuits, we are  only required to play a number of rounds that scales like $\cO(\epsilon^{-2})$. Moreover, Bob never wins if Alice plays mirror descent, $\epsilon=\Omega(\log(n)^{-1})$ and the distinguishing functions can be computed efficiently.

From a high-level perspective, Theorem~\ref{thm:finite rounds} below states that if Bob can always  efficiently find distinguishing functions and they can be computed efficiently, then Alice can also find and sample from a high-temperature Gibbs state that is close to the ideal distribution. Note that in this work we do not come up with such a strategy for Bob, but rather explore the consequences of the existence of such a strategy to understand the limitations and connections of verification and simulation of random quantum circuit experiments. 
\begin{thm}[Alice approximately learns $\nu$ after $\epsilon^{-2}$ rounds for random circuits, informal version of Thm.~\ref{thm:distinguishingandsimulating}]\label{thm:finite rounds}
Let $\nu$ be the probability distribution of the outcome of a random quantum circuit on $n$ qubits stemming from a $2^{-2n-1}$-approximate two design and $\epsilon>0$ be given. Suppose that Bob succeeds in providing functions $f_1,\ldots,f_T$ that can be computed in polynomial time and distinguish $\nu$ from a sequence $\mu_1,\ldots,\mu_T$ of distributions that Alice discloses, with $T$ the maximal number of rounds at most $\cO(\epsilon^{-2})$. Then there is an algorithm that allows Alice
to learn a distribution $\mu_{T+1}$ exploiting the revealed information on $f_t$ that can be
sampled from in time $e^{\cO(\epsilon^{-1})}$. Moreover, $\mu_T$ is $\epsilon$ close in total variation distance to $\nu$.
\end{thm}

Note that we can efficiently sample from the output distribution as long as $\epsilon=\Omega(\log(n)^{-1})$.
We will prove this result in Thm.~\ref{thm:distinguishingandsimulating}, but it is intimately connected to the fact that the output distributions of random quantum circuits are very "flat", as the probability of the outcomes is mostly of the order $2^{-n}$. For such distributions mirror descent converges very fast and we will see that they are well-approximated by high temperature Gibbs states. On the other hand, for the optimization problems like MAXCUT we expect good solvers to have outcome distributions that are highly concentrated on low energy strings. In contrast with very flat distributions, mirror descent converges slower for such concentrated distributions.

An important corollary of our theorem is that either the hardness conjectures of random quantum circuits are not valid for distances $\epsilon=\Omega(\log(n)^{-1})$ or a complete certification strategy for Bob, providing a discrimination function for every guess of Alice, is impossible with polynomial resources. Given the wide range of results that establish the hardness of sampling from random quantum circuits requiring slightly stronger assumptions than our result, we believe that our results indicate that efficient and scalable certification of random circuits in terms of empirical averages of functions is not possible.

\subsubsection{Distinguishing from uniform would invalidate HOG conjecture}
Our second result concerns a connection between the hardness of fooling the XHOG problem and distinguishing the output of a random circuit from the uniform distribution. We refer to Sec.~\ref{sec:hoguniform} for a precise definition of the XHOG problem.
There we also show that if the conjectures related to the hardness of fooling the XHOG problem are true, not even distinguishing from the uniform distribution in polynomial time should be possible. Thus, although it might be possible that Bob can only efficiently distinguish during the first rounds, before mirror descent converges as in Thm~\ref{thm:finite rounds}, this suggests that even the first round might be difficult to win if we restrict to efficient strategies.
Thus, Bob will have to resort to verification strategies that take super-polynomial time to demonstrate a quantum advantage within our game. More precisely, in Section~\ref{sec:hoguniform} we prove the following result.
\begin{thm}[Distinguishing from uniform and HOG, informal version of Thm.~\ref{thm:disthog}]
Let $f:\{0,1\}^{n}\to\{0,1\}$ be a function that for some $\epsilon>0$ satisfies:
    $\mathbb{E}_\nu(f)-\mathbb{E}_{\mathcal{U}}(f)\geq \epsilon,$
where $\nu$ is the outcome distribution of a random quantum circuit stemming from a $2^{-2n-1}$-approximate
two design in $n$ qubits. Then there is an algorithm that samples from a distribution that fools XHOG up to $\epsilon$ using $\cO(\epsilon^{-2})$ evaluations of $f$ in expectation. 
\end{thm}

One possible criticism of the above framework is that it might be in general hard to distinguish the outcome of any circuit stemming from a random ensemble of circuits from the uniform distribution.
However, this is not true, as we show in Appendix~\ref{app:stabilizerstates} that in case Bob is sampling from a randomly generated stabilizer circuit, Alice can easily fool XHOG.

\subsubsection{Effects of hardware errors}
All the results above concern the outcome distribution of the \emph{ideal} circuit. In Sec.~\ref{sec:noisydevices} we extend our results to the approximate simulability of the outcome distribution of noisy devices. We show that, under doubly-stochastic noise, the number of rounds of the verification game when Alice uses mirror descent decreases exponentially with the depth of the circuit Bob is implementing.
As we believe these results are interesting beyond quantum advantage proposals and apply to the broader topic of classical simulability of noisy circuits, we state them in more general terms. Below we state a specialized version of our  main result regarding the complexity of approximating the statistics of outcomes of noisy circuits, Theorem~\ref{thm:complexitynoise}:
\begin{thm}[Informal version of Thm.~\ref{thm:complexitynoise}]\label{thm:complexitynoise_intro}
Let $\nu$ be the outcome distribution of a noisy quantum circuit on $n$ qubits of depth $D$ affected by local depolarizing noise with rate $p$ after each gate, measured in the computational basis.
Given functions $f_1,\ldots,f_k:\{0,1\}^n\to [-1,1]$ and $\epsilon>0$, mirror descent will converge to a distribution $\mu_T$ satisfying:
\begin{align*}
\left|\mathbb{E}_\nu(f_i)-\mathbb{E}_{\mu_T}(f_i)\right|\leq \epsilon
\end{align*}
for all $1\leq i\leq k$ in at most $T=\cO(\epsilon^{-2}(1-p)^{2D+2}n)$ iterations. Moreover, we can sample from $\mu_T$ by evaluating $f_1,\ldots,f_k$ at most 
\begin{align*}
   \operatorname{exp}\lb \frac{4(1-p)^{2D+2}n}{\epsilon}\rb
\end{align*}
times.
\end{thm}

As shown in the recent~\cite{2009.05532}, which we discuss in more detail shortly, this restrains the power of noisy quantum computers to demonstrate a significant advantage versus classical methods for more structured problems.
Let us exemplify this with the noisy MAXCUT example:
\begin{example}[MAXCUT continued- noisy circuits]
Let us exemplify the consequences of Theorem~\ref{thm:complexitynoise_intro} to approximating MAXCUT on a noisy quantum device. Suppose that Bob's device suffers from local depolarizing noise with rate $p$ and consists of a circuit of depth $D$. In this scenario, Alice will be able to obtain an expected value of MAXCUT that is $\epsilon$ close to Bob's by sampling from the distribution $\mu_t$ given by:
\begin{align*}
\mu_t(x)=\frac{e^{\beta f_G(x)}}{\mathcal{Z}},\quad \mathcal{Z}=\sum_{x\in\{0,1\}^n}e^{\beta f_G(x)}
\end{align*}
with $\beta=\epsilon^{-1}(1-p)^{2D+2}n$. That is, the noise decreases the inverse temperature $\beta$ we have to go when performing classical simulated annealing to obtain comparable results. 
\end{example}
As is discussed in more detail in~\cite{2009.05532}, results like the one above can be used to rigorously establish maximal depths before noisy quantum devices are not outperformed by polynomial time classical algorithms. But the main message of the example above in our context of verification is that if there are clear candidate functions to distinguish the output of the noisy quantum circuit, then the noise will make it easier for Alice to simulate the output of the device, as one would expect.
We refer to Section~\ref{sec:noisydevices} for a derivation of this bound and a more detailed discussion of its consequences.

\section{Preliminaries}

We will now introduce some basic definitions and notation together with
the concepts of mirror descent and rejection sampling, which are relevant to our work.

\subsection{Notation}\label{sec:preliminaries}

\paragraph{Probability distributions on binary strings:} we define $F=\lb\{0,1\}^n\rb^{[-1,1]}$ to be the set of functions $f:\{0,1\}^n\to[-1,1]$. 

Given two probability distributions $\mu,\nu$ on $\{0,1\}^{n}$, we define their total variation distance $\|\nu-\mu\|_{TV}$ as:
\begin{align*}
    \|\nu-\mu\|_{TV}=\frac{1}{2}\sum\limits_{x\in \{0,1\}^{n}}\left|\mu(x)-\nu(x)\right|.
\end{align*}
Moreover, given a function $f:\{0,1\}^n\to\R$ we define
\begin{align*}
    \|f\|_{\infty}=\sup\limits_{x\in \{0,1\}^{n}}|f(x)|.
  \end{align*}

  \paragraph{Distinguishability measures for quantum states and unitary designs:} we are also going to need other distinguishability measures for distributions and quantum states. We will introduce them only for quantum states and note that the corresponding classical definition is obtained by considering the classical probability distribution as a diagonal quantum state. For two quantum states $\rho,\sigma\in\M_{2^n}$ we define their relative entropy to be:
  \begin{align*}
      S(\rho||\sigma)=\tr{\rho\lb\log\rho-\log\sigma\rb}
  \end{align*}
  if $\text{kern }\rho\subset\text{kern }\sigma$ and $+\infty$ otherwise.
  Moreover, we define the $\alpha-$R\'enyi entropies $S_\alpha$  for $\alpha>1$ to be given by:
  \begin{align*}
      S_\alpha(\rho)=-\frac{1}{\alpha-1}\log\lb\tr{\rho^{\alpha}}\rb.
  \end{align*}
  and the von Neumann entropy to be $S_1(\rho)=S(\rho)=-\tr{\rho\log(\rho)}$.
  Note that we have:
  \begin{align*}
     n\geq S(\rho)\geq S_\alpha(\rho).
  \end{align*}
  
  Let us also set our notation and terminology for random quantum circuits. Given a distribution $\tau$ on the unitary group on $n$ qubits, $U(2^n)$, and $t\in\N$, we define $\mathcal{G}^{(t)}_{\tau}:\M_{2^{tn}}\to\M_{2^{tn}}$ to be the quantum channel:
  \begin{align*}
  \mathcal{G}^{(t)}_{\tau}(X)=\int\limits_{U(2^n)}U^{\otimes t} X\lb U^\dagger\rb^{\otimes t}d\tau.
  \end{align*}
  $\mathcal{G}^{(t)}_{\tau}$ is then said to be an $\epsilon$-approximate $t$-design~\cite{Ambainis_2007} if
  \begin{align*}
  \|\mathcal{G}^{(t)}_{\tau}-\mathcal{G}^{(t)}_{\mu_G}\|_{\diamond}\leq\epsilon,
  \end{align*}
  where $\|\cdot\|_{\diamond}$ is the diamond norm and $\mu_G$ is the Haar measure on the unitary group. Moreover, given $C$ distributed according to $\tau$, we will always denote by $\nu$ the probability measure we obtain by measuring $C\ket{0}$ in the computational basis, i.e.
  \begin{align*}
  \nu(x)=\tr{\ketbra{x}C\ketbra{0}^{\otimes n}C^\dagger}
  \end{align*}
  for $x\in \{0,1\}^n$.

\subsection{Distinguishing distributions}\label{sec:disti_distri}

The total variation is widely accepted as one of the most natural and operationally relevant measures of distinguishability for two probability distributions. One of the reasons for that is its dual formulation. One can show that:
\begin{align}\label{equ:definition_tv}
    \|\nu-\mu\|_{TV}=\frac{1}{2}\sup\limits_{f\in F}(\mathbb{E}_\mu(f)-\mathbb{E}_\nu(f)).
\end{align}
Thus, it quantifies by how much the expectation values of two functions can differ on the two distributions.
Moreover, defining $S=\{x\in \{0,1\}^{n}:\mu(x)\geq\nu(x)\}$ and letting $\chi_S$ be the indicator function of $S$, it is easy to see that:
\begin{align*}
    \|\nu-\mu\|_{TV}=\frac{1}{2}(\mathbb{E}_\mu(\chi_S-\chi_{S^c})-\mathbb{E}_\nu(\chi_S-\chi_{S^c})).
\end{align*}

That is, we can restrict to differences of indicator functions in Eq.~\eqref{equ:definition_tv}. 

The total variation distance also has an operational interpretation in terms of distinguishability of two distributions. Indeed, consider the scenario in which with probability $\tfrac{1}{2}$ we are given a sample from $\mu$ and with probability $\tfrac{1}{2}$ we are given a sample from $\nu$. Then once can show that the optimal probability of guessing correctly from which distribution the sample came from is given by 
\begin{align*}
p_{\text{guess}}=\frac{1}{2}[1+\|\mu-\nu\|_{TV}].
\end{align*}  
Furthermore, the optimal strategy consists of responding $\mu$ if the sample was in $S$ and $\nu$ otherwise. Thus, we see that the total variation distance naturally allows us to quantify the distinguishability of two distributions in the one-shot setting. However, if we have access to $m$ samples of the distribution instead of one and have to distinguish them, then the success probability is then $\|\mu^{\otimes m}-\nu^{\otimes m}\|_{TV}$.

The characterization given in Eq.~\eqref{equ:definition_tv} can also be used in yet another way to distinguish probability distributions given access to multiple samples.
Suppose we have a witness function $f$ that the total variation distance between $\mu$ and $\nu$ is at least $\epsilon$, i.e. 
\begin{align}
    \left|\mathbb{E}_\mu(f)-\mathbb{E}_\nu(f))\right|\geq\epsilon.
\end{align}
We can then use the empirical average w.r.t. to $f$ to distinguish the distributions. To see why, given samples $X_1,\ldots,X_s$ from $\mu$ and $Y_1,\ldots,Y_s$ from $\nu$, it follows from Hoeffding's inequality that:
\begin{align}\label{equ:conclusion_hoeefding}
    \left|s^{-1}\sum\limits_{i=1}^sf(X_i)-\mathbb{E}_\mu(f)\right|\leq\frac{\epsilon}{2},\quad \left|s^{-1}\sum\limits_{i=1}^sf(Y_i)-\mathbb{E}_\nu(f)\right|\leq\frac{\epsilon}{2}
\end{align}
with probability of at least $1-\delta$ as long as 
\begin{align}\label{equ:samplehoeffding}
s=\cO(\epsilon^{-2}\log(\delta^{-1})). 
\end{align}
Thus, we can check if the empirical average of the samples is closer to $\mathbb{E}_\mu(f)$ or $\mathbb{E}_\nu(f)$ and use this as criterium to chose from which distribution we think the samples came from. 
A simple application of the triangle inequality demonstrates that this strategy will succeed with probability at least $1-\delta$. Thus, we conclude from Eq.~\eqref{equ:samplehoeffding} and the discussion above that as long as $\epsilon^{-2}\log(\delta^{-1})=\cO(\text{poly}(n))$, polynomially many samples and evaluations of the function $f$ are sufficient to certify that two distributions are at least at a certain distance $\epsilon$ in total variation and distinguish them. Of course, this in no sense discards the possibility that finding the distinguishing function $f$ itself or evaluating it may not be possible in polynomial time.

The discussion above allows us to estimate the number of samples we need to provide at each round of the game to ensure that the probability of a deviation greater than $\epsilon$ from the target is upper bounded by $1-\cO(\delta)$ for some $\delta$.
As at each round $t$ we have to estimate $t$ expectation values up to $\epsilon$, obtaining $\cO(\epsilon^{-2}\log(t \delta^{-t}))$ samples for each round ensures that the probability one of them deviates by more than $\epsilon$ is at most $\delta^t$. By a union bound, the probability that there was a deviation after $T$ rounds is at most 
\begin{align*}
    \sum\limits_{t=1}^T\delta^t\leq \frac{\delta}{1-\delta}=\cO(\delta)
\end{align*}
for $\delta\leq \frac{1}{3}$. Thus, letting the number of samples per round grow like $t\log(t \delta^{-1})$ is enough to ensure that the probability of an error occurring at some point remains of order $\delta$.

\subsubsection{Discussion on generality of the model}

This set of strategies to distinguish probability distributions is closely related to the statistical queries model~\cite{Kearns_1998,stat_query}. In this model to learn or distinguish distributions, one is not given access to samples from a distribution $\nu$. Rather, one is given access to an oracle that is also specified by a precision parameter $\epsilon>0$. When queried with a function $f\in F$, the oracle returns an estimate $e_f$ satisfying $|e_f-\mathbb{E}_\nu(f)|\leq \epsilon$. Thus, in some sense we can say that in our game it is Bob's task to distinguish his distribution from Alice's in the statistical query model. However, as discussed in more detail in Appendix~\ref{sec:gen_limitations}, some of our results generalize to the case where the distinguishing functions $f$ do not act on one sample, but rather a block of samples.

Note, however, that this is not the most general model to distinguish two probability distributions efficiently given samples. Indeed, one could consider more generally the scenario where we are given polynomially many samples of the distribution and can act on all of them simultaneously with a function that can be computed in polynomial time. Proving the impossibility of distinguishing two distributions in such a scenario is a daunting task, as discussed in more detail in Appendix~\ref{sec:gen_limitations}, and is out of reach of the results of this manuscript. Nevertheless, we stress that our results do apply to the strategies encountered in the literature.

\subsection{Mirror descent}

Mirror descent with the von Neumann entropy as potential~\cite{Bubeck2015, tsuda_matrix_2005} is an optimization and learning algorithm to approximate probability distributions efficiently and in a structured way through a series of Gibbs probability measures. 
It allows us to formally connect the problem of distinguishing probability distributions and learning them.
That is, given some target distribution  $\nu$ on $\{0,1\}^n$ we wish to learn, say the output distribution of a given random quantum circuit, mirror descent is an iterative procedure that proposes a sequence of $\mu_0,\ldots,\mu_T$ of guesses for $\nu$. Furthermore, the initial distribution $\mu_0$ is the uniform distribution $\mathcal{U}$.
The algorithm requires us to find functions $f_1,\ldots,f_T:\{0,1\}^{n}\to[-1,1]$ that allow us to distinguish $\mu_t$ from $\nu$, i.e.
\begin{align}\label{equ:distinguishabilityft}
\mathbb{E}_{\mu_t}(f_{t+1})-\mathbb{E}_{\nu}(f_{t+1})\geq\epsilon
\end{align}
for some given distinguishability parameter $\epsilon>0$. 
Note that we may assume without loss of generality that the equation in \eqref{equ:distinguishabilityft} holds without the absolute value, as if the inequality holds in the reverse direction we can just pick $-f_t$ instead.
One can now appreciate the direct connection between mirror descent and our verification game.
Of course, it is a priori not clear how to find such functions in a traditional mirror descent application. In the certification game, this problem is overcome by having the responsibility to provide $f$ 
on Bob's side.
Also note that if no function $f_t$ exists that satisfies~\eqref{equ:distinguishabilityft}, then $\|\nu-\mu_t\|_{TV}\leq\epsilon$ by the dual formulation of the total variation distance in eq.~\eqref{equ:definition_tv}.

\begin{algorithm}[tp!]
\caption{\textit{Mirror descent for learning probability distributions\label{eq:algomirror}.}
}
\label{alg:HUtomo}
\begin{algorithmic}[1]
\Function{Mirror descent}{$T,\epsilon$}
\State{Set   $\mu_0=\mathcal{U}$}
\Comment initialize to the uniform distribution
\For{$t=1,\ldots,T=\lceil 8\rl{\nu}{\mathcal{U}}\epsilon^{-2}\rceil$}\label{eq:firstforloop}
\State{Demand function $f_{t+1}$ such that $\mathbb{E}_{\mu_t}(f_{t+1})-\mathbb{E}_{\nu}(f_{t+1})\geq\epsilon$ }
\If{Given $f_t$}
\State{Set $\mu_{t+1}(x)= \operatorname{exp}(-\frac{\epsilon}{4} \sum\limits_{i=1}^{t+1}f_i(x))/ \mathcal{Z}_{t+1}$.} \Comment{ Update the guess.}
\EndIf
\If{ no such function exists}
\State{Return $\mu_t$}\label{line:didnotfind}
\State{\textbf{break loop}}
\EndIf
\EndFor
\State{Return $\mu_T$ and \textbf{exit function}}\label{line:exitedend} \Comment{Current guess is  $\epsilon$ indistinguishable from $\nu$}
\EndFunction
\end{algorithmic}
\end{algorithm}

As outlined in Algorithm~\ref{eq:algomirror}, mirror descent works by updating the probability measure as:
\begin{align}\label{eq:Aliceguess}
\mu_{t+1}= \operatorname{exp}\lb-\frac{\epsilon}{4} \sum\limits_{i=1}^{t+1}f_i\rb/ \mathcal{Z}_{t+1},
\end{align}
where 
\begin{align*}
\mathcal{Z}_{t+1}=\sum\limits_{x\in\{0,1\}^n}\operatorname{exp}\lb-\frac{\epsilon}{4} \sum\limits_{i=1}^{t+1}f_i(x)\rb
\end{align*}
is the partition function.
As we update the distributions, the candidate distributions $\mu_t$ become closer and closer to the target distribution, as made precise by the following lemma:
\begin{lem}\label{lem:mirrordesc}
The distributions $\mu_t$ of the algorithm~\ref{alg:HUtomo} satisfy:
\begin{align}\label{equ:relativeentropydecay}
    \rl{\nu}{\mu_t}\leq -t\frac{\epsilon^2}{8}+\rl{\nu}{\mathcal{U}},
\end{align}
where $\mathcal{U}$ is the uniform distribution.
\end{lem}
Eq.~\eqref{equ:relativeentropydecay} is a standard property of mirror descent~\cite{Bubeck2015}. We give a simplified proof and discuss  basic properties of this algorithm in Appendix~\ref{sec:mirror_descent_basics}. Also note that in principle we can "recycle" distinguishing functions. That is, if Alice updates her guess a few times, it could be the case that her distribution $\mu_t$ does not satisfy 
\begin{align}\label{equ:distinguishabilitypast}
\mathbb{E}_{\mu_t}(f_i)-\mathbb{E}_{\nu}(f_i)\leq\epsilon
\end{align}
for some previously disclosed $f_i$. In this case, she can update in terms of $f_i$ again until all previous expectation values also agree. This version of the algorithm is given in Algorithm~\ref{alg:HUtomo2} of Appendix~\ref{sec:mirror_descent_basics}.

Exploiting the direct connection between the verification protocol and mirror descent, we can directly use
lemma~\ref{lem:mirrordesc} to bound the number of rounds of the game in terms of $\rl{\nu}{\mathcal{U}}$. We then immediately obtain:
\begin{thm}\label{cor:highentropy}
The output of algorithm~\ref{alg:HUtomo} satisfies:
\begin{align}\label{equ:outputsmallintv}
\|\mu_t-\nu\|_{TV}\leq \epsilon
\end{align}
after at most $T\leq \lfloor 8\epsilon^{-2}\rl{\nu}{\mathcal{U}}\rfloor+1$ iterations.
\end{thm}
\begin{proof}
If we break the algorithm at Line~\eqref{line:didnotfind}, then, by the variational formulation of the total variation distance we have that Eq.~\eqref{equ:outputsmallintv} holds. 

To see that this must happen after at most $\lfloor 8\epsilon^{-2}\rl{\nu}{\mathcal{U}}\rfloor+1$ steps, note that by Eq.~\eqref{equ:relativeentropydecay} we have the relation
\begin{align}\label{equ:totalvariationperiteration}
0\leq \rl{\nu}{\mu_t}\leq\rl{\nu}{\mathcal{U}}-t\frac{\epsilon^2}{8}.
\end{align}
Thus, a total number of iterations $T$ that is larger than $\lfloor 8\epsilon^{-2}\rl{\nu}{\mathcal{U}}\rfloor+1$ would contradict the positivity of the relative entropy.
\end{proof}
Note that Eq.~\eqref{equ:totalvariationperiteration} ensures that we make constant progress in relative entropy at each iteration of the algorithm. Theorem~\ref{cor:highentropy} implies that if Bob provides a sequence of functions $f_1,\ldots,f_{T+1}$ that allow distinguishing $\nu$ from the sequence $\mu_0,\ldots,\mu_T$ of at most $\cO(\epsilon^{-2}\rl{\nu}{\mathcal{U}})$ distributions up to an error $\epsilon$, then we can also find a distribution that is $\epsilon$ close to it in total variation distance. Furthermore, as we will show later, it is possible to use this connection to the relative entropy to quantify the effect of noise on the complexity of learning the distribution.

\subsection{Rejection sampling}
Let us now show one way how to generate samples from $\mu_t$ and the underlying complexity. We will use the standard technique of rejection sampling described in Algorithm~\ref{alg:rejection}.

We refer to~[Appendix B.5]\cite{levin2017markov} for a brief review of its properties. In rejection sampling we sample indirectly from a target distribution $\mu_t$ by first generating a sample $x$ from an easy to sample distribution 
$\gamma(x)$ and accepting the sample with probability $\mu_t(x)/(M\gamma(x))$, where $M$ is a constant such that the ratio is bounded by $1$.
It is a standard fact that rejection sampling will output a sample from $\mu_t$ after $M$ runs in expectation, as the probability of rejection follows a geometric distribution with parameter $M^{-1}$.
\begin{algorithm}[b!]
\caption{\textit{Rejection sampling.}
}\label{alg:rejection}
\begin{algorithmic}[1]
\Require{ability to generate samples from distribution $\gamma$ on $\{0,1\}^n$, distribution $\mu_t$ on $\{0,1\}^n$, constant $M$ such that $\frac{\mu_t(x)}{M\gamma(x)}\leq 1$, ability to compute $\frac{\mu_t(x)}{M\gamma(x)}$ and samples from $\mathcal{U}([0,1])$.}
\Function{Rejection sampling}{M}

\State{Sample $u$ distributed according to $\mathcal{U}([0,1])$ and $x$ distributed acoording to $\gamma$.}
\If{$u\leq \frac{\mu_t(x)}{M\gamma(x)}$}
\State{Output $x$}
\EndIf
\EndFunction
\end{algorithmic}
\end{algorithm}
In the case of Gibbs distributions $\mu_t= \operatorname{exp}(-H_t)/\mathcal{Z}_t$, where
\begin{align*}
H_t(x)=\frac{\epsilon}{4}\sum\limits_{i=1}^tf_i(x),
\end{align*}
a common choice for $\gamma$ is the uniform distribution and $M_t=\frac{2^n }{\cZ_t}$, where $\cZ_t$ is once  again the partition function.
Note that for this choice of $M_t$, we have that:
\begin{align*}
\frac{\mu_t(x)}{M_t\cU(x)}=e^{-H_t(x)}\leq1,
\end{align*}
as we may assume without loss of generality that  $H_t(x)\geq 0$ for all $x\in\{0,1\}^n$. In particular, note that with this choice, we never have to compute the partition function $\cZ_t$ to run rejection sampling, only $H_t(x)$. Thus, we conclude that the complexity of running one round of rejection sampling is the same as computing $H_t(x)$. Let us now estimate how many rounds are required in expectation before we accept a sample:
\begin{lem}[Sampling from $\mu_t$]\label{lem:rejectionsamp}
Let $\mu_t$ be the guess at iteration $t$ of Algorithm~\ref{alg:HUtomo}.
Running rejection sampling with the uniform distribution as $\gamma$ and $M_t=\frac{2^n }{\cZ_t}$ returns a sample from $\mu_t$ after at most $e^{\frac{\epsilon}{4} t}$ trials and evaluations of $H_t$, in expectation.
\end{lem}
\begin{proof}
Note that by our previous discussion the expected number of trials is $M_t=\frac{2^n }{\cZ_t}$.
By construction, 
$f_t$ are functions with image $[-1,1]$. Thus, it follows from a triangle inequality that:
\begin{align*}
\|H_t\|_{\infty}\leq\frac{\epsilon}{4}\sum\limits_{i=1}^t\|f_i\|_\infty\leq \frac{t\epsilon}{4}.
\end{align*}
This implies that
\begin{align*}
    \cZ_t=\sum\limits_{x\in\{0,1\}^n} \operatorname{exp}(-H_t(x))\geq 2^n e^{-\frac{\epsilon}{4} t}.
\end{align*}
We conclude from the last inequality that
\begin{align*}
M_t=\frac{2^n }{\cZ_t}\leq e^{\frac{\epsilon}{4} t}
\end{align*}
which yields the claim.
\end{proof}
We see that as long as $\epsilon t=\cO(\log(n))$, then we can sample from $\mu_t$ in a polynomial expected number of trials and evaluations of $H_t$.

In practice, rejection sampling is not necessarily the most efficient way of simulating probability distributions and other techniques to sample from a Gibbs distribution such as Glauber dynamics or simulated annealing~\cite{levin2017markov}  perform better. However, rejection sampling allows for a simple analytical analysis, which is more challenging for more refined techniques.

\section{Random quantum circuits}\label{sec:simulability}

Let us now discuss the implications to the verification of quantum advantage proposals based on sampling the output distribution of random circuits. The key technical assumption behind various state-of-the-art classical hardness proofs for quantum advantage proposals is the property that the underlying ensemble is an approximate two design~\cite{Hangleiter_2018}. Thus, we will also depart from this assumption.
We then have:
\begin{lem}\label{lem:entropycondition}
Let $\nu$ be the probability distribution of the outcome of a random quantum circuit on $n$ qubits stemming from a $2^{-2n-1}$-approximate two design. Then, with probability at least $1-\delta$, we have:
\begin{align*}
S(\nu)\geq n-\left[ \log(1/\delta)+\log (3) \right].
\end{align*}
\end{lem}
\begin{proof}
We refer to Prop.~\ref{prop:lowershannon} in Appendix~\ref{app:lowerbounds} for a proof.
\end{proof}
Similar statements were shown in~\cite{PhysRevLett.122.210502,boson_far_uniform}. From this, we have:
\begin{thm}[Distinguishing output distributions and classical simulability]\label{thm:distinguishingandsimulating}
Let $\nu$ be the probability distribution of the outcome of a random quantum circuit on $n$ qubits stemming from a $2^{-2n-1}$-approximate two design and $\epsilon>0$ be given. Suppose that we can distinguish $\nu$ from a sequence $\mu_1,\ldots,\mu_T$ of distributions that can be sampled from in polynomial time. Moreover, we can distinguish them by functions $f_1,\ldots,f_T$ that can be evaluated in polynomial time. That is:
\begin{align}\label{equ:distinguishateachstep}
\forall 1\leq t\leq T:\mathbb{E}_{\mu_{t}}(f_{t})-\mathbb{E}_{\nu}(f_{t})\geq\epsilon.
\end{align}
with $f_{t}$ computable in polynomial time.
Then, with probability at least $1-\delta$, we can find distributions $\mu_t$ satisfying:
\begin{align}\label{equ:outputgood}
\|\nu-\mu_t\|_{TV}=\sqrt{2\lb\log(3)+\log(1/\delta)-t\frac{\epsilon^2}{8}\rb}
\end{align}
and sample from it in time $\cO(e^{\frac{c}{\epsilon}}\operatorname{poly}(n))=\cO(\operatorname{poly}(n))$. In particular, if $T=\cO(\epsilon^{-2})$, then the output distribution $\mu_T$ will also be $\epsilon$ close in total variation distance to the target.
\end{thm}

\begin{proof}
As stated in Lemma~\ref{lem:entropycondition}, we have that with probability at least $1-\delta$
\begin{align*}
S(\nu)\geq n-\left[\log(3)+\log(1/\delta)\right].
\end{align*}

Conditioned on the event above, it  follows from Thm.~\ref{cor:highentropy}  that mirror descent outputs a distribution satisfying Eq.~\eqref{equ:outputgood} after at most 
\begin{align*}
8\epsilon^{-2}\lb \log(3)+\log(1/\delta)\rb
\end{align*}
iterations, or equivalently after that many game rounds. 
Now, at each iteration $t$ of mirror descent, we need a function satisfying Eq.~\eqref{equ:distinguishateachstep}.
Moreover, we have that $\mu_t\propto\operatorname{exp}\lb -\epsilon/4\sum\limits_i  f_i \rb$. Thus, if all the $f_i$ can be computed in polynomial time, then it follows from Lemma ~\ref{lem:rejectionsamp} that we can also sample from $\mu_t$ using rejection sampling in polynomial time. This is because Lemma~\ref{lem:rejectionsamp} implies that we need at most
\begin{align*}
\operatorname{exp}\lb \frac{2(\log(3)+\log(1/\delta))}{\epsilon}+\tfrac{\epsilon}{4}\rb %
\end{align*}
rejection sampling rounds, in expectation.
As each round of rejection sampling requires us to evaluate the functions $f_i$ once and we suppose that they can be evaluated in polynomial time, this gives the claim.
\end{proof}

Therefore, if Bob provides for every proposed distribution $\mu_t$ of Alice a polytime computable function $f_t$ that distinguishes it from $\nu$, after at most a constant number of 
rounds Alice will be sampling efficiently from an approximate distribution. It is interesting to point out that the certification game and the sampling of Alice remains efficient, even if we relax the condition of Lemma \ref{lem:entropycondition} to $S(\nu)\geq n-\cO(\log(n))$ or request $\epsilon$ to decrease with the size $n$ of the quantum device with scaling $\epsilon=\cO(\log(n)^{-1})$. 

A direct consequence of our result is that if the hardness conjecture of random quantum circuits is true,  then Bob must fail to provide a certification function that is efficiently computable 
at some stage of the certification game. 
A natural question would then be to ask at which stage Bob will fail to provide such a function.
In the following section, we will prove that if the XHOG conjecture~\cite{aaronson2019classical}
is correct, Bob must fail at the first round of the game, i.e., even distinguishing the output distributions from the uniform distribution cannot be done in polynomial time.

We note that these results have important differences from the results in~\cite[Appendix 11]{boson_far_uniform}. There the authors show the existence of a high min-entropy distribution that can be sampled from classically and is indistinguishable from the random quantum circuit by classical circuits of polynomial size. 
This is because in our case we have the guarantee of a good approximation in total variation distance, i.e. the distributions are indistinguishable under any function after a couple of iterations. Another difference is that, given the distinguishing functions, our framework allows for finding the probability distribution that approximates the random circuit. Moreover, if the distinguishing functions are given and can be computed efficiently, then the outcome distribution can also be sampled from efficiently. However, to the best of our knowledge, the aforementioned result does not give an algorithm to find such an approximate distribution. Finally, our framework allows us to work with the Shannon entropy instead of the min-entropy. The min-entropy is notoriously more difficult to bound and always smaller than the Shannon entropy.

\section{Distinguishing from the uniform distribution}\label{sec:hoguniform}

To the best of our knowledge, the state-of-the-art approach for the verification of quantum advantage proposals based on random circuit sampling is the linear cross-entropy heavy output generation problem (XHOG)~\cite{aaronson2019classical}, which is closely related to the linear cross-entropy benchmarking (linear XEB) fidelity $\mathcal{F}_{\operatorname{XEB}}$~\cite{Neill_2018,Boixo_2018,Arute2019}. The XHOG refers to the problem of, given some circuit $C$, generating distinct samples $z_1,\ldots,z_k$ such that:
\begin{align}\label{equ:highenoughprob}
    \mathbb{E}_{i}\left[\left|\left\langle z_{i}|C| 0^{n}\right\rangle\right|^{2}\right] \geq b / 2^{n}
\end{align}
for some $b>1$ with probability at least $s=\frac{1}{2}+\frac{1}{2b}$ and $k$ satisfying:
\begin{align}\label{equ:numbersamplesxhhog}
    k \geq \frac{1}{((2 s-1) b-1)(b-1)}.
\end{align}
Note that, given the samples $z_1,\ldots,z_k$, verifying that they indeed satisfy eq.~\eqref{equ:highenoughprob} requires us to compute the probability of the outcomes under the ideal circuit. 
In turn, the linear cross-entropy fidelity for a distribution $\mu$, $\mathcal{F}_{\operatorname{XEB}}(\mu)$, as defined in~\cite{Neill_2018,Boixo_2018,Arute2019}, is given by:
\begin{align}\label{equ:cross_entropy}
    \mathcal{F}_{\operatorname{XEB}}(\mu)=2^n\mathbb{E}_{\mu}\lb \nu\rb-1,
\end{align}
where we interpreted the probability distribution $\nu$ as a function on $\{0,1\}^n$ that outputs the corresponding probability $\nu(x)$.
The linear cross-entropy can also be formulated as 
\begin{align*}
 \mathcal{F}_{\operatorname{XEB}}(\mu)=2^n\mathbb{E}_{\mu}\lb f_\nu\rb-1,
\end{align*}
where $f_\nu$ is given by $f_\nu(x)=2^n\nu(x)$. Although such a function does not immediately fit our framework, as it may take values higher than $1$, in Appendix \ref{app:heavyoutputgeneration_function} we show how to approximate it by a bounded function. The underlying intuition is that as $\nu$ is very flat for random quantum circuits, for very few inputs the function $f$ will take values that are not of constant order. Thus, as long as the measure $\mu$ is not too concentrated on strings of high value, it is possible to cut-off the function $f$ without significantly changing the expectation value.

A simple manipulation then shows that samples $z_i$ from $\mu$ satisfy 
\begin{align*}
    \mathbb{E}_{i}\left[\left|\left\langle z_{i}|C| 0^{n}\right\rangle\right|^{2}\right] \geq \frac{1+\mathcal{F}_{\operatorname{XEB}}(\mu)}{ 2^{n}}.
\end{align*}
In~\cite{aaronson2019classical}, the authors relate the complexity of computing outcome probabilities of random quantum circuits to the XHOG problem. More precisely, the authors start by assuming that there is no polynomial-time classical algorithm that takes as input a (random) quantum circuit $C$ and
produces an estimate $p$ of $p_0 = \mathbb{P}[C \operatorname{ outputs } 0]$ such that
\begin{align*}
    \mathbb{E}\left[\left(p_{0}-2^{-n}\right)^{2}\right]=\mathbb{E}\left[\left(p_{0}-p\right)^{2}\right]+\Omega\left(2^{-3 n}\right).
\end{align*}
where the expectations are taken over circuits $C$ as well as the algorithm’s internal randomness.
They then show that this conjecture implies that there is no polynomial time algorithm that solves the XHOG problem. As the verification of XHOG requires us to estimate the outcome probabilities, this path to proving  and verifying the hardness of the sampling task also implies that XHOG is not verifiable in polynomial time.

But the hardness of XHOG imposes barriers to even more basic verification tasks.
As we will see now, the hardness of XHOG would imply that it is impossible to efficiently distinguish the output from the circuit from the uniform distributions. In turn, efficient distinguishability implies a polynomial time algorithm to spoof XHOG. Therefore, Bob will have to resort to verification strategies that take superpolynomial time to demonstrate a quantum advantage in our game.
It should not be surprising that fooling XHOG is related to distinguishing from the uniform distribution, as both tasks require us to identify higher-than-average probability strings. What is particular to the case of random circuits is the fact that distinguishing implies we can also sample from a distribution that fools XHOG. The proof of the statement will once again rely on the fact that the output distribution is essentially flat. We will start by showing that any distribution with large $2-$R\' enyi entropy cannot have small sets of large mass in the following sense:
\begin{lem}\label{lem:highcollisionsmallmass}
Let $\nu$ be a probability distribution on $\{0,1\}^{n}$ such that $S_2(\nu)\geq n-\log(c)$ for some constant $c>0$. Then, for any $\epsilon>0$ and subset $L\subset \{0,1\}^{n}$ we have that 
\begin{align*}
\nu(L)=\sum\limits_{x\in L}\nu(x)\geq\epsilon
\end{align*}
implies that
\begin{align*}
    |L|\geq \epsilon^2 c^{-1} 2^n.
\end{align*}
\end{lem}
\begin{proof}
Note that the condition $S_2(\nu)\geq n-\log(c)$ is equivalent to
\begin{align}\label{equ:small2norm}
\sum\limits_{x\in\{0,1\}^n}\nu(x)^2\leq c2^{-n}.
\end{align}
Moreover, we have
\begin{align}\label{equ:boundonsizel}
    \sum\limits_{x\in\{0,1\}^{n} }\nu(x)^2\geq \sum\limits_{x\in L }\nu(x)^2\geq  \frac{\epsilon^2}{|L|}.
\end{align}
To see the last inequality, note that, by the concavity of the function $x\mapsto x^2$, 
\begin{align*}
   \frac{1}{|L|} \sum\limits_{x\in L }\nu(x)^2\geq \lb \frac{1}{|L|}\sum\limits_{x\in L }\nu(x)\rb^2 \geq \frac{\epsilon^2}{|L|^2}.
\end{align*}
Combining~\eqref{equ:boundonsizel} with~\eqref{equ:small2norm} we conclude that:
\begin{align*}
    \frac{\epsilon^2}{|L|}\leq c2^{-n},
\end{align*}
which yields the claim after rearranging the terms.
\end{proof}

It then immediately follows that:
\begin{lem}\label{lem:from_dist_tohog}
Let $f:\{0,1\}^{n}\to\{0,1\}$ be a function that for some $\epsilon>0$ satisfies:
\begin{align}\label{equ:disting_app}
    \mathbb{E}_\nu(f)-\mathbb{E}_{\mathcal{U}}(f)\geq \epsilon,
\end{align}
where $\nu$ is the outcome distribution of a random quantum circuit stemming from a $2^{-2n-1}$-approximate two design on $n$ qubits. Let
\begin{align*}
    L=\{x\in\{0,1\}^n:f(x)=1\}.
\end{align*}
Then, with probability at least $1-\delta$,
\begin{align}\label{equ:Lisbig2}
     |L|=\Omega\lb \epsilon^2\delta 2^n\rb.
\end{align}
and
\begin{align}\label{equ:generateshog2}
    \frac{1}{|L|}\nu(L)\geq  \frac{1}{2^n}+\frac{\epsilon}{|L|}\geq \frac{1+\epsilon}{2^n}.
\end{align}
\end{lem}
\begin{proof}
First note that Eq. \eqref{equ:disting_app} is equivalent to
\begin{align*}
    \nu(L)\geq \frac{|L|}{2^n}+\epsilon
\end{align*}
and, in particular, $\nu(L)\geq \epsilon$.
Moreover, Eq. \eqref{equ:generateshog2} readily follows by dividing the equation above by $|L|$.
As we show in Prop.~\ref{prop:lowershannon} in Eq.~\eqref{equ:normdesign} that with probability at least $1-\delta$
\begin{align}\label{equ:small_collision}
    S_2(\nu)\geq n-\log(3)+\log(\delta). 
\end{align}
Conditioned on Eq. \eqref{equ:small_collision}, it follows from Lemma \ref{lem:highcollisionsmallmass} that 
\begin{align}\label{equ:Lisprettybig}
|L|\geq c\epsilon^2 \delta 2^n
\end{align}
for some constant $c>0$, which yields Eq. \eqref{equ:Lisbig2}.
\end{proof}

We restricted the result above to functions with binary outputs to simplify the arguments, but we show in Appendix \ref{app:hogapp} that this can be done without loss of generality. That is, given a function $f$ that distinguished the distributions for $\epsilon$ and range $[-1,1]$, there always exists some $f'$ that is binary and has the same properties and distinguishes the distribution up to $\tfrac{\epsilon^2}{17}$.

If the function $f$ in Lemma~\ref{lem:from_dist_tohog} can computed in polynomial time, then we can use it to fool XHOG in polynomial time:
\begin{thm}[From distinguishing to fooling XHOG]\label{thm:disthog}
Let $f$ as in Lemma~\ref{lem:from_dist_tohog} for some $\epsilon>0$. Moreover, let $\mathcal{U}(L)$ be the uniform distribution on $L$. Then we can sample from $\mathcal{U}(L)$ by evaluating $f$ a total of $\cO(\epsilon^{-2})$ many times, in expectation. Moreover, samples from $\mathcal{U}(L)$ violate HOG up to $\epsilon$.
\end{thm}
\begin{proof}
Let us start with the statement that samples from $\mathcal{U}(L)$ violate XHOG up to at least $\epsilon$.
To see this, note that:
\begin{align*}
    \mathbb{E}_{\nu}\lb\mathcal{U}(L)\rb=\sum\limits_{x\in L}\frac{\nu(x)}{|L|}\geq \frac{1+\epsilon}{2^n}
\end{align*}
by Eq.~\eqref{equ:generateshog2}, where with some abuse of notation we see $\mathcal{U}(L)$ as a function that outputs the probability of $x$ under $\mathcal{U}(L)$ given $x$. To sample from $\mathcal{U}(L)$, we can once again resort to a variation of rejection sampling. We sample a point $x_1\in\{0,1\}^n$ from the uniform distribution and compute $f(x_1)$. If $f(x_1)=1$, we output $x_1$. If not, we rerun this procedure with a new sample $x_2$. It is easy to see that when we accept, $x_i$ is uniformly distributed on $L$. Moreover, the probability of accepting is $\frac{L}{2^n}$. By Eq.~\eqref{equ:Lisprettybig}, it then follows that the probability of accepting is at least $\Omega(\epsilon^{2})$, from which we obtain that the expected number of rounds is $\cO(\epsilon^{-2})$. As for each round we have to evaluate $f$ once, the claim follows.
\end{proof}
It follows that if we can efficiently distinguish the output from the uniform distribution, then we can fool XHOG. In particular, it would follow from the conjectures of~\cite{aaronson2019classical} that it is not possible to distinguish the output of random circuits from the uniform distribution for $\epsilon$ at least inverse polynomial in $n$ in polynomial time if XHOG cannot be solved in polynomial time. 

Note that the results of this section only required the property that the underlying random circuit ensemble is a two design. However, it is well-known that random Clifford unitaries also are two designs and can still be simulated efficiently classically. 
It is then natural to ask if finding a function distinguishing the output  of a random Clifford from the uniform distribution can be found in polynomial time. As we show in Appendix~\ref{app:stabilizerstates}, for random Cliffords, the function can be found and be computed in polynomial time. 

\section{Noisy devices}\label{sec:noisydevices}
Most of the current implementations of quantum circuits have high levels of noise.  In principle, highly correlated noise can make simulating the quantum device classically even more complex~\cite{Kliesch_2011}. However, in other contexts, as demonstrated for boson sampling~\cite{Qi_2020}, noise can render the simulation significantly less complex.
Here, we show that in the  scenario of doubly stochastic Markovian noise, i.e., quantum channels that map the maximally mixed state to itself, the noise diminishes the complexity of approximating the underlying distribution being sampled from.
To see what we mean, let us go back the game described in Section~\ref{sec:summary}. There we considered Bob to have a quantum advantage and win the game if he could find functions $f_1,\ldots,f_t$ whose expectation values under his distribution Alice cannot reproduce classically. 
However, suppose now that Bob's device is affected by noise, say a global depolarizing channel with parameter $0\leq p\leq 1$. One would then expect that as $p\to 1$, it should be easier for Alice to reproduce the statistics of Bob's device and win the game.
As we show below, if Alice once again uses mirror descent to find her candidate distributions, then the algorithm will converge faster to a distribution that reproduces the statistics of Bob's distribution as the level of noise increases. For example, for MAXCUT, this will also make it easier for her to  sample from the distribution mirror descent proposes, as a smaller number of iterations implies not having to go down to lower temperatures.

\subsection{Strong data processing inequalities and mirror descent}

Let us first recall the following definition:
\begin{defi}[Strong data processing inequality]
A doubly stochastic quantum channel $T:\M_{d}\to\M_d$ (i.e. a CPTP map satisfying $T(\one)=T^*(\one)=\one$) is said to satisfy a strong data processing inequality with contraction $\alpha>0$ w.r.t. $\frac{I}{d}$ if for all states $\rho$ we have:
\begin{align*}
    \rl{T(\rho)}{\frac{\one}{d}}\leq (1-\alpha)\rl{\rho}{\frac{\one}{d}}.
\end{align*}
\end{defi}
Strong data processing inequalities for doubly stochastic can be derived using the framework of hypercontractivity and logarithmic Sobolev inequalities~\cite{Kastoryano_2013,M_ller_Hermes_2016,M_ller_Hermes_2016_doubly,2011.05949,Beigi_2020} and are known explicitly in some cases. See also~\cite{1909.02383} for another approach. As shown by other works~\cite{aharonov1996limitations,ben2013quantum,Muller_Hermes_2015,2009.05532}, strong data processing can be used to quantify how useful a noisy quantum device is and for how long it can sustain interesting computations.

Let us illustrate the strong data processing inequality with an example. Let $\Phi_p:\M_2\to\M_2$ be the depolarizing channel on one qubit with depolarizing parameter $p$, i.e.
\begin{align*}
    \Phi_p(\rho)=(1-p)\rho+p\frac{\one}{2}.
\end{align*}
The authors of~\cite{M_ller_Hermes_2016} show that:
\begin{align*}
    \rl{\Phi_p^{\otimes n}(\rho)}{\frac{\one}{2^n}}\leq\lb 1-2p+p^2\rb\rl{\rho}{\frac{\one}{2^n}}
\end{align*}
and similar results are available for other relevant noise models. In particular, optimal inequalities have been derived in~\cite{Kastoryano_2013,M_ller_Hermes_2016_doubly} for tensor products of single qubit, doubly stochastic channels. 

Suppose now that the noisy circuit of interest consists of $n$ qubits initialized to $\ket{0}^{\otimes n}$, $D$ layers of unitaries $U_1,U_2,\ldots,U_D$ and measurement in the computational basis. However, due to imperfections in the implementation, the state initialization, the measurements and the unitaries are noisy. We will model this by assuming that every layer of unitaries is proceeded by a layer of doubly stochastic quantum channel $\Phi$ that satisfies a strong data processing inequality $\alpha$. 
Moreover, we will assume that the measurement is also affected by an extra noisy channel $\Phi$. We only make these assumptions to simplify the notation and argument, but it is easy to adapt the argument to different channels at different times and different noise rates.
Under the assumptions above, the probability distribution $\nu$ describing the outcomes of the noisy device is given by:
\begin{align}\label{equ:noisyoutput}
    \nu(i)=\tr{\ketbra{i}T(\ketbra{0}^{\otimes n})}=\tr{\ketbra{i}(M\circ \Phi\circ\mathcal{U}_D\circ\Phi\circ\ldots\circ \mathcal{U}_1 \circ \Phi(\ketbra{0}^{\otimes n})},
\end{align}
where $\mathcal{U}_i$ is the channel given by conjugations with $U_i$ and $M$ is the q.c. channel
\begin{align*}
    M(\rho)=\sum\limits_{i=0}^{2^n-1}\tr{\rho\ketbra{i}}\ketbra{i}
\end{align*}
and 
\begin{align*}
    T(\ketbra{0}^{\otimes n})=(M\circ \Phi\circ\mathcal{U}_D\circ\Phi\circ\ldots\circ \mathcal{U}_1 \circ \Phi(\ketbra{0}^{\otimes n}).
\end{align*}
We then have that if we want to approximate the values of $k$ functions, then the number of iterations of mirror descent $T$ required to achieve that decreases exponentially with the noise level. More precisely:
\begin{thm}\label{thm:complexitynoise}
Let $\nu$ be the distribution defined in Eq.~\eqref{equ:noisyoutput}, $\epsilon>0$ and assume $\Phi$ satisfies a strong data processing inequality with parameter $\alpha$. Given functions $f_1,\ldots,f_k:\{0,1\}^n\to [-1,1]$, mirror descent will converge to a distribution $\mu_t$ satisfying:
\begin{align*}
\left|\mathbb{E}_\nu(f_i)-\mathbb{E}_{\mu_t}(f_i)\right|\leq \epsilon
\end{align*}
for all $1\leq i\leq k$ in at most $T=\cO(\epsilon^{-2}(1-\alpha)^{D+1}n)$ iterations. Moreover, we can sample from $\mu_t$ by evaluating $f_1,\ldots,f_k$ at most 
\begin{align*}
   \operatorname{exp}\lb \frac{4(1-\alpha)^{D+1}n}{\epsilon}\rb
\end{align*}
times.
\end{thm}

\begin{proof}
Mirror descent will converge to a distribution with the desired properties after $8\epsilon^{-2}\rl{\nu}{\mathcal{U}}$ iterations, see Prop.~\ref{prop:convergence_mirror} for a proof.
Moreover, by Lemma~\ref{lem:rejectionsamp}, the complexity of sampling from $\mu$ is bounded by 
\begin{align*}
    \operatorname{exp}\lb \frac{4\rl{\nu}{\mathcal{U}}}{\epsilon}\rb
\end{align*}
evaluations of the functions $f_i$.
Thus, the statement follows if we can bound the relative entropy of the outcome.
Note that:
\begin{align*}
    \rl{\nu}{\mathcal{U}}=\rl{T(\ketbra{0}^{\otimes n})}{\frac{\one}{2^n}},
\end{align*}
as the maximally mixed state gives rise to the uniform distribution when measured in the computational basis.
By the data processing inequality:
\begin{align*}
    \rl{T(\ketbra{0}^{\otimes n})}{\frac{\one}{2^n}}\leq \rl{\Phi\circ\mathcal{U}_D\circ\Phi\circ\ldots\circ \mathcal{U}_1 \circ \Phi(\ketbra{0}^{\otimes n})}{\frac{\one}{2^n}}
\end{align*}
and by our assumption on $\Phi$:
\begin{align*}
    \rl{\Phi\circ\mathcal{U}_D\circ\Phi\circ\ldots\circ \mathcal{U}_1 \circ \Phi(\ketbra{0}^{\otimes n})}{\frac{\one}{2^n}}\leq (1-\alpha)\rl{\mathcal{U}_D\circ\Phi\circ\ldots\circ \mathcal{U}_1 \circ \Phi(\ketbra{0}^{\otimes n})}{\frac{\one}{2^n}}.
\end{align*}
The relative entropy is unitarily invariant, thus:
\begin{align*}
    (1-\alpha)\rl{\mathcal{U}_D\circ\Phi\circ\ldots\circ \mathcal{U}_1 \circ \Phi(\ketbra{0}^{\otimes n})}{\frac{\one}{2^n}}=(1-\alpha)\rl{\Phi\circ\ldots\circ \mathcal{U}_1 \circ \Phi(\ketbra{0}^{\otimes n})}{\frac{\one}{2^n}}.
\end{align*}
Applying the chain of arguments above another $D$ times we conclude that:
\begin{align*}
    \rl{T(\ketbra{0}^{\otimes n})}{\frac{\one}{2^n}}\leq (1-\alpha)^{D+1}\rl{\ketbra{0}^{\otimes n}}{\frac{\one}{2^n}}=(1-\alpha)^{D+1}n,
\end{align*}
which yields the claim.
\end{proof}
Thus, we see that the complexity of approximate sampling from the output of noisy circuits decreases exponentially with the noise level and depth. 
For instance, for circuits with local depolarizing noise with parameter $p$ and depth $D$, the theorem above gives a complexity of:
\begin{align}\label{equ:localdeplogsobo}
   \operatorname{exp}\lb \frac{4(1-p)^{2D+2}n}{\epsilon}\rb.
\end{align}
evaluations of the distinguishing functions. Thus, we see that our framework has the desirable feature that it becomes easier for Alice to mock Bob's device as the noise increases. Unfortunately, the scaling with $n$ in the bound above is undesirable for high-entropy distributions. Thus, we plan to derive more specialized bounds in upcoming work.

\section{Acknowledgments}
DSF was supported by VILLUM FONDEN via the QMATH Centre of Excellence under Grant No. 10059. RGP was supported by the Quantum Computing and Simulation Hub, an EPSRC-funded project, part of the UK National Quantum Technologies Programme. We thank Anthony Leverrier and Juani Bermejo-Vega for helpful comments and discussions. 

\bibliographystyle{plainnat}
\bibliography{mybib}

\appendix
\section{Basic properties of mirror descent with the Shannon entropy as potential}\label{sec:mirror_descent_basics}
In this section we review some basic properties of the mirror descent algorithm with the Shannon entropy as potential. We start with a simple proof of the update rule behind mirror descent in Lemma~\ref{lem:update-rule}. It shows how to obtain a Gibbs state $\tau_1$ that is closer in relative entropy to a target distribution $\nu$ departing from a Gibbs state $\tau_0$ and a function that distinguishes the latter from the target distribution.
\renewcommand{\tr}{\operatorname{tr}}

\begin{lem} \label{lem:update-rule}
Let $\nu$ be a probability distribution on $n$ bits and $\epsilon>0$. Fix a function $H_0:\{0,1\}^n\to\R$ and let $\tau_0= e^{-H_0}/\mathcal{Z}_0$. 
Suppose that for some other (bounded) function $f:\{0,1\}^n\to[-1,1]$ we have:
\begin{align}\label{equ:lower_energy}
\mathbb{E}_{\tau_0}(f)-\mathbb{E}_{\nu}(f)\geq \epsilon.
\end{align}
Set $H_{1}=H_0+\frac{\epsilon}{4}f$. Then, the Gibbs state $\tau_1= e^{-H_1}/\mathcal{Z}_1$ obeys
\begin{equation*}
S(\nu \| \tau_{1}) - S(\nu \|\tau_0) \leq- \frac{\epsilon^2}{8}
\end{equation*}
\end{lem}

\begin{proof}
We have:
\begin{align}\label{equ:first-step}
&S(\nu\| \tau_{1}) - S (\nu \| \tau_0) 
= \mathbb{E}_\nu\left[  H_{1}-H_0\right]
+ \ln \left( \frac{\mathcal{Z}_1}{\mathcal{Z}_0} \right)
\end{align}
By construction, $H_1-H_{0}=\frac{\epsilon}{4}f$ and the first term equals $\frac{\epsilon}{4} \mathbb{E}_{\nu}(f)$. 
The logarithmic ratio can be bounded using Jensen's inequality: 
\begin{align*}
 \ln \left( \frac{\mathcal{Z}_1}{\mathcal{Z}_0}\right) =-\ln  \left(\mathbb{E}_{\tau_1}\left[e^{\frac{\epsilon}{4}f}\right]\right)\leq -\mathbb{E}_{\tau_1}\left[\frac{\epsilon}{4}f\right],
\end{align*}
as the function $-\log(x)$ is convex.
It then follows that
\begin{align}\label{equ:Peierls-Bogoliubov}
    & \mathbb{E}_\nu\left[  H_{1}-H_0\right]
+ \ln \left( \frac{\mathcal{Z}_1}{\mathcal{Z}_0} \right)\leq \frac{\epsilon}{4}\left(\mathbb{E}_\nu\right[f\left]-\mathbb{E}_{\tau_1}\left[f\right]\right).
\end{align}
We will now show that the expectation values of $f$ on the distributions $\tau_1$ and  $\tau_0$ are $\cO(\epsilon)$ close. We can  then replace the expectation value over $\tau_1$ by  only paying a  small price and then use our hypothesis in Eq.~\eqref{equ:lower_energy}.

To that end, a direct computation shows that
\begin{align}
\rl{\tau_0}{\tau_1}=\frac{\epsilon}{4}\mathbb{E}_{\tau_0}(f)+\ln\left(\frac{\mathcal{Z}_1}{\mathcal{Z}_0}\right).
\end{align}
As in Eq.~\eqref{equ:Peierls-Bogoliubov}, we then can estimate the second term by:
\begin{align}\label{equ:bound_relative_entropy}
\rl{\tau_0}{\tau_1}=\frac{\epsilon}{4}\mathbb{E}_{\tau_0}(f)+\ln\left(\frac{\mathcal{Z}_1}{\mathcal{Z}_0}\right)\leq  \frac{\epsilon}{4}\left(\mathbb{E}_{\tau_0}(f)-\mathbb{E}_{\tau_1}(f)\right).
\end{align}
By Pinsker's inequality:
\begin{align*}
\left(\mathbb{E}_{\tau_0}(f)-\mathbb{E}_{\tau_1}(f)\right)^2\leq 2\rl{\tau_0}{\tau_1}\leq \frac{\epsilon}{2} \left(\mathbb{E}_{\tau_0}(f)-\mathbb{E}_{\tau_1}(f)\right).
\end{align*}
As $\left(\mathbb{E}_{\tau_0}(f)-\mathbb{E}_{\tau_1}(f)\right)\geq0$ (see Lemma~\ref{lem:monotone_decreasing} below for a proof) this yields
\begin{align*}
\left(\mathbb{E}_{\tau_0}(f)-\mathbb{E}_{\tau_1}(f)\right)\leq \frac{\epsilon}{2}.
\end{align*}
We then see that:
\begin{align*}
\frac{\epsilon}{4}\left(\mathbb{E}_\nu\right[f\left]-\mathbb{E}_{\tau_1}\left[f\right]\right)\leq \frac{\epsilon}{4}\left(\mathbb{E}_\nu\left[f\right]-\mathbb{E}_{\tau_0}\left[f\right]+\frac{\epsilon}{2}\right).
\end{align*}
By our assumption in Eq.~\eqref{equ:lower_energy} we may then bound the right hand side in Eq.~\eqref{equ:Peierls-Bogoliubov} by $-\frac{\epsilon^2}{8}$ and finally obtain:
\begin{align*}
   S(\nu\| \tau_{1}) - S (\nu \| \tau_0)=\mathbb{E}_\nu\left[  H_{1}-H_0\right]
+ \ln \left( \frac{\mathcal{Z}_1}{\mathcal{Z}_0} \right)\leq-\frac{\epsilon^2}{8}.
\end{align*}
The claim follows.
\end{proof}
With this Lemma at hand, we can then show that mirror descent will converge to a probability distribution approximating the expectation values of a given set of functions with respect to a probability measure. We will now show that mirror descent allows for recovering the expectation values of a set of functions and also give guarantees as to how well we approximate the distribution globally.

\renewcommand{\tr}[1]{\text{tr}\lb#1\rb}

\begin{algorithm}[tp!]
\caption{\textit{Mirror descent for reproducing expectation values\label{eq:algomirror_expectation}.}
}
\label{alg:HUtomo2}
\begin{algorithmic}[1]
\Require{Expectation value of functions $f_1,\ldots,f_k:\{0,1\}^n\to[-1,1]$  with respect to probability measure $\nu$ on $n$ bits.}
\Function{Mirror descent}{$T,\epsilon$}
\State{Set   $\mu_0=\mathcal{U}$}
\Comment initialize to the uniform distribution
\For{$t=1,\ldots,T=\lceil 8\rl{\nu}{\mathcal{U}}\epsilon^{-2}\rceil$}\label{eq:firstforloop2}
\State{Check if $\left|\mathbb{E}_{\mu_t}(f_{i})-\mathbb{E}_{\nu}(f_{t})\right|\leq\epsilon$ for all $1\leq i\leq k$.}
\If{Given that for a $f_{i}$ we have $\left|\mathbb{E}_{\mu_t}(f_{i})-\mathbb{E}_{\nu}(f_{t})\right|\geq\epsilon$}
\If{$\mathbb{E}_{\mu_t}(f_{i})-\mathbb{E}_{\nu}(f_{t})\geq\epsilon$}
\State{Set $\mu_{t+1}(x)= \operatorname{exp}(-\frac{\epsilon}{4}f_i(x)+\log(\mu_t))/ \mathcal{Z}_{t+1}$.} \Comment{ Update the guess.}
\ElsIf{$\mathbb{E}_{\mu_t}(f_{i})-\mathbb{E}_{\nu}(f_{t})\leq-\epsilon$}
\State{Set $\mu_{t+1}(x)= \operatorname{exp}(\frac{\epsilon}{4}f_i(x)+\log(\mu_t))/ \mathcal{Z}_{t+1}$.} \Comment{ Update the guess.}
\EndIf

\ElsIf{ For all $\left|\mathbb{E}_{\mu_t}(f_{i})-\mathbb{E}_{\nu}(f_{t})\right|\leq\epsilon$ }
\State{Return $\mu_t$}\label{line:didnotfind2}
\State{\textbf{break loop}}
\EndIf
\EndFor
\State{Return $\mu_T$ and \textbf{exit function}}\label{line:exitedend2} \Comment{Current guess is  $\epsilon$ indistinguishable from $\nu$}
\EndFunction
\end{algorithmic}
\end{algorithm}
We then have:

\begin{prop}[Mirror descent converges]\label{prop:convergence_mirror}
Algorithm~\ref{eq:algomirror_expectation} returns a probability measure $\mu_t$ that satisfies
\begin{align}\label{equ:expectations_close}
\left|\mathbb{E}_{\mu_t}(f_{i})-\mathbb{E}_{\nu}(f_{t})\right|\leq\epsilon
\end{align}
for all $1\leq i\leq k$ after at most $t\leq \lceil 8\rl{\nu}{\mathcal{U}}\epsilon^{-2}\rceil$ iterations. Moreover, $\mu_t$ satisfies 
\begin{align}\label{equ:approximation_TV}
\|\mu_t-\nu\|_{TV}\leq \sqrt{2\left(\rl{\nu}{\mathcal{U}}-\frac{t\epsilon^2}{8}\right)}.
\end{align}
\end{prop}
\begin{proof}
If we exit the algorithm at line~\ref{line:didnotfind2}, then the output satisfies Eq.~\eqref{equ:expectations_close} by definition. Thus, it only remains to prove that this is indeed the case after $8\lceil\rl{\nu}{\mathcal{U}}\epsilon^{-2}\rceil$ iterations. But note that Lemma~\ref{lem:update-rule} and the update rules of Algorithm~\ref{eq:algomirror_expectation} ensure that we have:
\begin{align*}
S(\nu \| \mu_{t+1}) - S(\nu \|\mu_t) \leq- \frac{\epsilon^2}{8}
\end{align*}
for all $t$.
Applying a telescopic sum and our initial choice $\mu_0=\mathcal{U}$ we see that
\begin{align}\label{equ:decay_entropy}
\rl{\nu}{\mu_t}\leq \rl{\nu}{\mathcal{U}}-t\frac{\epsilon^2}{8}
\end{align}
and the claim on the number of iterations follows from the positivity of the relative entropy. The claim in Eq.~\eqref{equ:approximation_TV} follows from Eq.~\eqref{equ:decay_entropy}.
\end{proof}

Finally, let us prove for completeness the following standard fact that we used in the proof of Lemma~\ref{lem:update-rule}:
\begin{lem}\label{lem:monotone_decreasing}
Let $\tau_0=\frac{e^{-H_0}}{\mathcal{Z}_0}$ be an arbitrary Gibbs probability measure on $n$ bits and for a function $f:\{0,1\}^n\to[-1,1]$ define for $\lambda>0$ the Gibbs probability measure 
\begin{align*}
    \tau_{\lambda}=\frac{e^{-H_0-\lambda f}}{\mathcal{Z}_\lambda}.
\end{align*}
Then $g:\lambda\mapsto \mathbb{E}_{\tau_{\lambda}}(f)$ is a monotone decreasing function.
\end{lem}
\begin{proof}
The proof is quite standard and simple. It is easy to check that:
\begin{align*}
    \frac{d}{d\lambda}g(\lambda)=-(\mathbb{E}_{\tau_{\lambda}}(f^2)-\mathbb{E}_{\tau_{\lambda}}(f)^2),
\end{align*}
which is the negative of the variance of the function $f$ under $\tau_{\lambda}$. Thus, we clearly have that $\frac{d}{d\lambda}g(\lambda)\leq 0$ and the claim follows.

\end{proof}
\section{Distinguishing function from the linear cross-entropy}\label{app:heavyoutputgeneration_function}
At first sight, the current verification procedures for random circuits, the linear cross-entropy benchmark or the heavy output generation problem, do not readily fit into our framework. As we will see now, this is not defined as the expectation value of a bounded function.
Indeed, define the function
\begin{align*}
    f(x)=2^n\nu(x)-1
\end{align*} 
where $\nu(x)$ is the probability of string $x$ under $\nu$, the output of the ideal circuit. 
We have that the linear cross-entropy is given by
\begin{align}
    \mathcal{F}_{\operatorname{XEB}}(\mu)=\mathbb{E}_{\mu}\lb f\rb-1.
\end{align}

In principle, the function $f$ could take values between $[-1,2^n-1]$, whereas our framework required the distinguishing functions to take values in $[-1,1]$. However, we will now show that by suitably discarding high values of $f$ and restricting $\mu$ to a suitable set of distributions, we can massage the linear cross-entropy into our framework. That is, we will find a bounded function $f_r$ that approximates $\mathcal{F}_{\operatorname{XEB}}$. Once again, the main property required to show this is the fact that random quantum circuits have very flat outcome distributions. 

To prove our claims, we will first assume that the distribution of probabilities of the outcomes is well-approximated by a Porter-Thomas distribution, as explained in detail below. We refer to~\cite{Boixo_2018} for a justification of this assumption and numerical evidence of its validity. It is possible to obtain similar but weaker results departing from the assumption that the output is an approximate $3$-design. However, for the sake of conciseness, we will restrict to the Porter-Thomas distribution.

The Porter-Thomas assumption is an approximation of the probability that a given outcome string $x$ will have for a family of quantum circuits.
More specifically, it assumes that for all strings $x$, the random variable corresponding to the value of $\nu(x)$ under this family of random quantum circuits follows the density $\kappa$ given by
\begin{align}\label{equ:porterthomas}
\kappa(p)=2^ne^{-p2^n}.
\end{align}
It is not difficult to see that this distribution is highly concentrated around its mean, $2^{-n}$, and that its variance is $2^{-2n}$, as it corresponds to an exponential distribution with parameter $2^{-n}$.

We then have:
\begin{prop}\label{prop:dist_cross_entropy}
Let $\nu$ be the output of a random quantum circuit on $n$ qubits and assume that the density of outcomes is given by a Porter-Thomas distribution with parameter $2^{-n}$, as in Eq.~\eqref{equ:porterthomas}. For a parameter $r\geq 1$ define $f_r:\{0,1\}^n\to[-1,1]$ as 
\begin{align*}
f_r(x)=r^{-1}(\min\{2^n\nu(x),r\}-1).
\end{align*}
Then for another distribution $\mu$ satisfying for some constant $C>0$
\begin{align}\label{equ:bound_distribution}
\mu(x)\leq C\nu(x)
\end{align}
almost surely for all $x$ such that $\nu(x)\geq r2^{-n}$ we have:
\begin{align}\label{equ:expectations_close_xeb}
\mathbb{E}\left[\left| \mathcal{F}_{\operatorname{XEB}}(\mu)-r\mathbb{E}_\mu(f_r)\right|\right]\leq Ce^{-r}(2+r).
\end{align}
where the expectation is taken over the circuits.
\end{prop}
\begin{proof}
By the definition of $f_r(x)$ we have that 
\begin{align*}
\mathbb{E}\left[\left| \mathcal{F}_{\operatorname{XEB}}(\mu)-r\mathbb{E}_\mu(f_r)\right|\right]=
\mathbb{E}\left[\sum\limits_{x:\nu(x)\geq \frac{r}{2^n}}(2^n\nu(x)-r)\mu(x)\right].
\end{align*}
By our assumption on the distribution in Eq. \eqref{equ:bound_distribution} we have that
\begin{align*}
\mathbb{E}\left[\sum\limits_{x:\nu(x)\geq \frac{r}{2^n}}(2^n\nu(x)-r)\mu(x)\right]\leq C\mathbb{E}\left[\sum\limits_{x:\nu(x)\geq \frac{r}{2^n}}(2^n\nu(x)^2-r\nu(x))\right].
\end{align*}
Furthermore, by the assumption that the probability of the output strings follows a Porter-Thomas distribution, we conclude that:
\begin{align*}
\mathbb{E}\left[\sum\limits_{x:\nu(x)\geq \frac{r}{2^n}}(2^n\nu(x)^2-r\nu(x))\right]=
\int\limits_{\frac{r}{2^n}}^{+\infty}2^{2n}e^{-2^nx}(2^nx^2-rx)dx=e^{-r}(2+r),
\end{align*}
which yields the claim.
\end{proof}
Thus, we see that for distributions that do not differ too much from the distribution of the outcome of the random circuit in the sense of Eq. \eqref{equ:bound_distribution}, $\mathcal{F}_{\operatorname{XEB}}$ can be approximated in our framework.
By picking a cut off at $\log(\epsilon^{-2})$ we ensure that the difference between the truncated $f_r$ and $\mathcal{F}_{\operatorname{XEB}}$ only differ by $\cO(\epsilon)$. 

Let us now discuss the condition in Eq. \eqref{equ:bound_distribution} in more detail. 
The condition in Eq. \eqref{equ:bound_distribution} has a natural interpretation for the problem at hand. The goal of Prop.~\ref{prop:dist_cross_entropy} is to identify conditions under which $\mathcal{F}_{\operatorname{XEB}}$ is well-approximated by a bounded function. However, if a probability measure $\mu$ only satisfies Eq. \eqref{equ:bound_distribution} for large values of $C$, it means it assigns high probability outcomes of $\nu$ even more weight than $\nu$. This in turn will yield higher values for $\mathcal{F}_{\operatorname{XEB}}(\mu)$. However, for outcome distributions that are not strongly concentrated on heavy outcomes, we expect Eq. \eqref{equ:bound_distribution} to hold for moderate values of $C$. For instance, for the uniform distribution the condition holds with $C=1$. 

However, it is possible to construct distributions that converge to the true distribution in total variation and for which $\mathcal{F}_{\operatorname{XEB}}(\mu)$ diverges. At the same time, it is possible to construct distributions that are a constant distance away from the ideal distribution in total variation, do not satisfy Eq.~\eqref{equ:bound_distribution} and nevertheless satisfy $\mathcal{F}_{\operatorname{XEB}}(\mu)=\mathcal{F}_{\operatorname{XEB}}(\nu)$. We will give the explicit constructions of these distributions shortly. But they showcase that in principle there is no connection between $\mathcal{F}_{\operatorname{XEB}}(\mu)$ and the total variation distance between $\mu$ and $\nu$.

Both constructions will exploit the fact that the $\mathcal{F}_{\operatorname{XEB}}$ is unbounded, as expected. Indeed, if the benchmark we were using were bounded, then at least we can always conclude from a convergence in total variation distance that the expectation values also have to converge.

Thus, we believe that these examples showcase why we cannot expect that $\mathcal{F}_{\operatorname{XEB}}$ can always be captured in our framework. Whereas our framework is intimately connected to the two distributions being close in total variation distance, this is not the case for similar linear cross entropy.
This was also observed in~\cite{limitations_XEB}, where the authors give additional arguments why the linear cross entropy is not connected with the total variation distance or fidelity in general.

\begin{example}[Distributions close in total variation distance but diverging $\mathcal{F}_{\operatorname{XEB}}$]
To construct our examples, observe that it follows from the Porter-Thomas assumption in Eq.~\eqref{equ:porterthomas} that for some given $c_1<1$, we expect $2^{(1-c_1)n}$ strings to have probability at least $c_1n2^{-n}$. Indeed, the expected number of strings with probability at least $c_2n2^{-1}$ is:
\begin{align*}
    2^n\int_{c_1n2^{-n}}^{\infty}2^ne^{-p2^n}dp=2^{(1-c_1)n}.
\end{align*}
Now define 
\begin{align*}
B_{c_1}=\{x\in\{0,1\}^n:\nu(x)\geq c_1n2^{-n}\}
\end{align*}
and let $\mathcal{U}_{B_{c_1}}$ be the uniform distribution on $B_{c_1}$. Further define the distribution
\begin{align*}
\mu_{c_1}=\left(1-\frac{1}{\sqrt{n}}\right)\nu+\frac{1}{\sqrt{n}}\mathcal{U}_{B_{c_1}}.
\end{align*}
Clearly, $\|\mu_{c_1}-\nu\|_{TV}=\cO(n^{-\frac{1}{2}})$. However, $\mathcal{F}_{\operatorname{XEB}}(\mu_{c_1})=\Omega(\sqrt{n})$, as $\mathcal{F}_{\operatorname{XEB}}(\mathcal{U}_{B_{c_1}})=\Omega(n)$ by definition. 

\end{example}

\begin{example}[Distributions far away in total variation distance but $\mathcal{F}_{\operatorname{XEB}}$ is similar]
To construct this example, we will resort to the same distribution as above. Let $a=\mathcal{F}_{\operatorname{XEB}}(\mathcal{U}_{B_{c_1}})$. Again, by the definiition of $B_{c_1}$, $a\geq c_1n-1$.
Now let $\mu'_{c_1}$ be defined as 
\begin{align}
\mu'_{c_1}=(1-1/a)\mathcal{U}+\frac{1}{a}\mathcal{U}_{B_{c_1}}.
\end{align}
Using the fact that $\mathcal{F}_{\operatorname{XEB}}(\mathcal{U})=0$, by the linearity of $\mathcal{F}_{\operatorname{XEB}}$ we get that $\mathcal{F}_{\operatorname{XEB}}(\mu'_{c_1})=1$, which is the expected value for $\nu$. Thus the value of $\mathcal{F}_{\operatorname{XEB}}$ coincides for both distributions.
But a reverse triangle inequality together with the fact that $\|\nu-\mathcal{U}\|_{TV}=\Omega(1)$ shows that 
\begin{align}
\|\mu'_{c_1}-\nu\|_{TV}=\Omega(1).
\end{align}

\end{example}

In spite of the limitations of the $\mathcal{F}_{\operatorname{XEB}}$ showcased above, for the uniform distribution the situation is less complicated. As the uniform distribution corresponds to the guess in the first round of the game and deserves a detailed analysis, we will now directly compute by how much a suitably cut-off and normalized linear cross-entropy allows for distinguishing the output of the random quantum circuit from the uniform distribution.

\begin{prop}\label{equ:distinguish_uniform}
Let $f_r$ and $\nu$ as in the statement of Prop. \ref{prop:dist_cross_entropy} and $\mathcal{U}$ be the uniform distribution on $n$ bits. Then we have for $r\geq 1$:
\begin{align}\label{equ:cut_off_expectation}
\mathbb{E}\left[\mathbb{E}_\nu(f_r)-\mathbb{E}_{\mathcal{U}}(f_r)\right]=\frac{1-e^{-r}(1+2r)}{r}
\end{align}
where the first expectation value is taken over the random quantum circuits.
\end{prop}
\begin{proof}
The proof is similar to the last proposition. We have that:
\begin{align*}
&\mathbb{E}\left[\mathbb{E}_\nu(f_r)-\mathbb{E}_{\mathcal{U}}(f_r)\right]=\\
&\sum\limits_{x:\nu(x)\leq \frac{r}{2^n}}r^{-1}\nu(x)2^n\left(\nu(x)-\frac{1}{2^n}\right)+(1-r^{-1})\sum\limits_{x:\nu(x)> \frac{r}{2^n}}\left(\nu(x)-\frac{1}{2^n}\right).
\end{align*}
Taking the expectation, the first sum above translates to the integral
\begin{align}
r^{-1}\int\limits_{0}^{\frac{r}{2^n}}2^{2n}e^{-2^n x}2^nx\left(x-\frac{1}{2^n}\right)dx=r^{-1}\left(1-e^{-r}(1+r^2+r)\right),
\end{align}
whereas the second translates to
\begin{align}
    (1-r^{-1})\int\limits_{\frac{r}{2^n}}^{+\infty}2^{2n}e^{-2^n x}\left(x-\frac{1}{2^n}\right)dx=(1-r^{-1})re^{-r}.
\end{align}
Summing the two expressions yields the claim.
\end{proof}
It is not immediately obvious how to maximize the expression in Eq. \eqref{equ:cut_off_expectation} analytically, but numerically solving it we see that it is around $r\simeq 3.21$, for which we obtain a violation of $\simeq0.22$. As the expected total variation distance is $1/e\simeq 0.36$ \cite{Boixo_2018} under the Porter-Thomas assumption, we see that this function is not far from the optimal distinguishing function. Thus, Bob could propose the function $f_{3.21}$ to distinguish the distribution of his device and the uniform distribution in the ideal case. Of course, as evaluating $f_{3.21}$ requires us to compute outcome probabilities, this is not an efficient distinguishing function. But the results of this section showcase that the linear cross-entropy fits into our framework by introducing a suitable cut-off as long as the underlying distribution does not put too much additional weight on heavy outputs.

\section{Bound on the Shannon entropy from design property}\label{app:lowerbounds}
We will now show that the Shannon entropy of the output distributions of approximate two designs when measured in the computational basis is essentially maximal. This result is similar in spirit to those of \cite{boson_far_uniform,PhysRevLett.122.210502}.

\begin{prop}\label{prop:lowershannon}
Let $U$ be a $(2,\epsilon 2^{-2n-1})$ approximate unitary design and define $\nu$ as before. Then, with  probability at least $1-\delta$:
\begin{align*}
    S(\nu)\geq n-\log(2+\epsilon)-\log(\delta^{-1})
\end{align*}
\end{prop}
\begin{proof}
For a Haar random unitary and $x\in\{0,1\}^n$ we have that:
\begin{align*}
 \mathbb{E}\lb  \left|\bra{x}U\ket{0}\right|^2\rb =\frac{1}{2^n},\quad \mathbb{E}\lb\left|\bra{x}U\ket{0}\right|^4\rb=\frac{2}{2^n(2^n+1)}.
\end{align*}
Thus, for an approximate two design as above, we have that:
\begin{align*}
\mathbb{E}\lb\left|\bra{x}U\ket{0}\right|^4\rb\leq\frac{2}{2^n(2^n+1)}+\frac{\epsilon}{2^n(2^n+1)}.
\end{align*}
Recall that the $2$-Renyi entropy $S_2$ is defined as:
\begin{align*}
     S_2(\nu)=-\log\lb \sum\limits_{x}\nu(x)^2\rb
\end{align*}
and that $S(\nu)\geq S_2(\nu)$. Moreover, the function $-\log$ is convex. Thus, it follows from Jensen's inequality that:
\begin{align*}
    \mathbb{E}\lb -\log\lb \sum\limits_{x}\nu(x)^2\rb\rb\geq -\log \lb \mathbb{E}\left[\sum\limits_{x}\nu(x)^2\right]\rb.
\end{align*}
From the computations above, we have that:
\begin{align*}
    \frac{2+\epsilon}{(2^n+1)}\geq \mathbb{E}\left[\sum\limits_{x}\nu(x)^2\right],
\end{align*}
from which we readily obtain that:
\begin{align*}
    \mathbb{E}\lb -\log\lb \sum\limits_{x}\nu(x)^2\rb\rb\geq n-\log(2+\epsilon).
\end{align*}
It follows from Markov's inequality that
\begin{align}\label{equ:normdesign}
\mathbb{P}\lb \sum\limits_{x}\nu(x)^2\geq  \frac{2+\epsilon}{\delta 2^n}\rb\leq \delta,
\end{align}
from which the claim follows.
\end{proof}

\section{Auxiliary results for Section \ref{sec:hoguniform}}\label{app:hogapp}
In Section \ref{sec:hoguniform} we showed that being able to efficiently distinguish the probability distributions arising from sampling from random circuits from the uniform distribution implies the ability to fool the XHOG problem associated to the same class of circuits to a certain level.

But we assumed that the function that distinguished the distributions has binary outputs, although our framework for distinguishability allows for functions with outputs in $[-1,1]$. We now show that it is always possible to obtain an efficiently computable function with binary outputs from an efficiently function with image $[-1,1]$ that distinguishes the two probability distributions, at the expense of a smaller distinguishability power.
\begin{lem}
Let $f:\{0,1\}^{n}\to [-1,1]$ be a function that can be computed in polynomial time such that for some probability measure $\nu$ and $\epsilon>0$ we have that:
\begin{align*}
    \bE_\nu(f)-\bE_{\mathcal{U}}(f)\geq \epsilon.
\end{align*}
Then there exists a function $f':\{0,1\}^{n}\to \{0,1\}$ that can be computed in polynomial time such that:
\begin{align*}
    \bE_\nu(f')-\bE_{\mathcal{U}}(f')\geq \frac{\epsilon^2}{17}.
\end{align*}
\end{lem}
\begin{proof}
First, we will consider instead of $f$ the shifted and normalized function
\begin{align*}
    \tilde{f}=\frac{f+1}{2},
\end{align*}
as it then has image in $[0,1]$ and clearly 
\begin{align*}
    \bE_\nu(\tilde{f})-\bE_{\mathcal{U}}(\tilde{f})\geq \frac{\epsilon}{2}.
\end{align*}
Assume w.l.o.g. that $\epsilon= m^{-1}$ for some integer $m$ and consider the discretization of $f_1$ of $f$ given by:
\begin{align*}
    f_1(x)=(8m)^{-1} \lceil 8m\tilde{f}(x)\rceil.
\end{align*}
Note that $\|\tilde{f}-f_1\|_{\infty}\leq (8m)^{-1}$ and, thus,
\begin{align}\label{equ:diff_expectation_values_discretized}
    \bE_\nu(f_1)-\bE_{\mathcal{U}}(f_1)\geq \frac{\epsilon}{4}.
\end{align}
Recall that for every real valued random variable $X$ we have:
\begin{align}\label{equ:expectation_levels_prob_distri}
    \bE(X)=\int \bP(X\geq x)dx,
\end{align}
Moreover, note that $f_1$ only takes the $8m+1$ possible values $\{0,(8m)^{-1},\ldots,1\}$. Thus. combining this observation with Eq.~\eqref{equ:diff_expectation_values_discretized} and the identity in Eq.~\eqref{equ:expectation_levels_prob_distri} we see that
\begin{align}\label{equ:difference_levels}
   \bE_\nu(f_1)=\sum\limits_{k=0}^{8m}\nu\lb f_1(x)\geq \frac{k}{8m} \rb\geq\frac{\epsilon}{4}+ \sum\limits_{k=0}^{8m}\frac{\left| \{f_1(x)\geq \frac{k}{8m}\}\right|}{2^n}.
\end{align}
It then follows that for at least one $0\leq k_0\leq 8m$ we have that:
\begin{align}\label{equ:functionf0violation}
    \nu\lb f_1(x)\geq \frac{k_0}{8m} \rb\stackrel{(1)}{\geq} \frac{\left| \{f_1(x)\geq \frac{k_0}{8m}\}\right|}{2^n}+\frac{\epsilon}{2(8m+1)}\geq \frac{\left| \{f_1(x)\geq \frac{k_0}{8m}\}\right|}{2^n}+\frac{\epsilon^2}{17},
\end{align}
because if the opposite inequality would hold in $(1)$ for all $k_0$ we would obtain a contradiction to Eq.~\eqref{equ:difference_levels} by summing over all $k_0$.
Thus, by setting $f'(x)=1$ if $ f_1(x)\geq \frac{k}{8m}$ and $0$ else and recalling that $m^{-1}=\epsilon$, we see that Eq. \eqref{equ:functionf0violation} immediately implies
\begin{align*}
    \bE_\nu(f')-\bE_{\mathcal{U}}(f')\geq \frac{\epsilon^2}{17},
\end{align*}
which yields the claim.
\end{proof}

\section{Generalizations and limitations of our results }\label{sec:gen_limitations} 

Let us discuss more precisely to what class of verification and distinguishability algorithms our results apply to, how it can be generalized and how it relates to the statistical query model well-known in statistical learning theory~\cite{Kearns_1998,stat_query}.

\subsection{The one-shot model}

Most current verification and distinguishability proposals we are aware of in the random circuit literature have a simple structure. They consist of defining a (not necessarily bounded) function like the cross-entropy benchmark and evaluating its empirical average.
We call this scenario one-shot as, despite involving an estimation over many samples, its theoretical analysis involve the distribution of a single realization.

As explained in Sec.~\ref{sec:disti_distri}, the optimal probability of success for correctly distinguishing two distributions from one sample is bounded for \emph{any} algorithm by the total variation distance. Thus, if we have that the trace distance between two distributions $\mu,\nu$ is small, then we can conclude that no algorithm will be able to perform significantly better than random guessing in the one-shot setting.

\subsubsection{Comparison to statistical query model}
This approach can be naturally cast in the statistical query model~\cite{Kearns_1998,stat_query}. In that model, one is not given access to samples of a distribution $\mu$, but one is allowed to query $\mathbb{E}_{\mu}(f)$ up to some additive error tolerance $\epsilon>0$ for arbitrary $f:\{0,1\}^n\to[-1,1]$.
We can see that it is possible to formalize our game in this model, as it is Bob's job to distinguish his distribution from Alice's through he expectaion values of such functions. And, as discussed before, this has a natural interpretation in terms of the succes probability of one-shot distinguishing algorithms.

\subsection{A more general framework: multiple copies discrimination}

However, there is no a-priori reason why we should limit ourselves to the one-shot scenario. %
We could define more generally efficient distinguishing algorithms that take as input a polynomial number of samples $m$ from a distribution, perform a polynomial-time postprocessing of the data and then outputs a guess.

In this general framework one considers the success probability of arbitrary distinguishing algorithms that take as an input a polynomial number of samples, i.e., $\mu^{\otimes m}$ instead of $\mu$ for $m=\textrm{poly}(n)$. 
Unfortunately, our techniques cannot discard the existence of such an efficient distinguishing procedure. Indeed, if there is a polynomial time function $f$ such that for some $m\geq 1$ we have that:
\begin{align}\label{equ:f_efficient}
    \left|\mathbb{E}_{\mu^{\otimes m}}(f)-\mathbb{E}_{\nu^{\otimes m}}(f)\right|=\Omega(n^{-m}),
\end{align}
then we can take $\cO(n^{2m}\log(\delta^{-1}))$ samples from the distribution, compute the empirical average of $f$ on them and distinguish the two with probability of success at least $1-\delta$. As $f$ can be computed in polynomial time, the empirical average can also be computed efficiently and this yields and efficient distinguishing procedure. And our techniques do not discard the existence of an efficient $f$ as in Eq.~\eqref{equ:f_efficient}.

Although in the main text we only considered the case of $m=1$ for our results, it should be noted that our results naturally extend to the settig in which $m\epsilon^{-1}=\cO(\log(n))$. That is, if we are in regime $\epsilon=\Omega(1)$, we can also consider the case in which the distinguishing functions act on a logarithmic number of samples. To see why, note that the proof of Thm.~\ref{thm:distinguishingandsimulating} relied solely on the fact that with probability at least $1-\delta$ we have
\begin{align}
S(\nu\|\mathcal{U})=\cO(1+\log(\delta^{-1}))
\end{align}
for the output distributions of approximate $2$-designs.

If we consider instead $m$ copies i.i.d. samples of the distribution $\nu$, the joint output distribution satisfies
\begin{align*}
S(\nu^{\otimes m}\|\mathcal{U}^{\otimes m})=mS(\nu\|\mathcal{U})=\cO(m(1+\log(\delta^{-1})))
\end{align*}
by the additivity of the relative entropy. Thus, as in Thm.~\ref{thm:distinguishingandsimulating}, if we set our error tolerance to be $\epsilon$, mirror descent will converge after $\cO(m^{2}\epsilon^{-2})$ iterations to a distribution $\mu$ that approximates $\nu^{\otimes m}$ up trace distance $\epsilon$. Moreover, sampling from $\mu$ using rejection sampling takes  $\cO(e^{m\epsilon^{-1}})$ evluations of $f$ on average. And from this we obtain that as long as the distinguishing functions are efficient and $m\epsilon^{-1}=\cO(\log(n))$, the whole procedure is efficient and our results still apply.

\section{Distinguishing the output distribution of stabilizer states}\label{app:stabilizerstates}
Note that we only assumed in the proofs of Sec.~\ref{sec:hoguniform} that the underlying circuit ensemble is an approximate two design. It is well-known that circuits being approximate two designs does not imply that one cannot sample from their output distribution efficiently, as is prominently exemplified by Clifford circuits. Random Cliffords are two designs~\cite{Dankert_2009,DiVincenzo_2002} and it is possible to simulate measurements in the computational basis efficiently for them~\cite{gottesman1998heisenberg,Aaronson_2004}. We now show that in this case, it is also possible to easily distinguish the outcome distribution from the uniform one, if this is possible at all.

It is not difficult to see that for a Pauli string $P=\otimes_{i=1}^n \sigma_{i}$ we have for a Clifford $C$ that
\begin{align*}
\tr{P C\ketbra{0}^{\otimes n}C^\dagger}\in\{-1,0,1\}.
\end{align*}
This is because, as Cliffords stabilize the Pauli group, we have that $\tilde{P}=C^\dagger P C$ is again, up to a global sign, a Pauli string. And for a Pauli string $\tr{\tilde{P}\ketbra{0}^{\otimes n}}\in \{0,1\}$. Now assume that there exists a Pauli string $P$ consisting only of $I$ and $Z$ Pauli matrices that differs from the identity and such that 
\begin{align*}
\tr{PC\ketbra{0}^{\otimes n}C^\dagger}\not=0.
\end{align*}
Then, interpreting diagonal operators as functions,  we have that the function $f=\frac{P+I}{2}$ is a binary function that satisfies:
\begin{align*}
|\mathbb{E}_\nu(f)-\mathbb{E}_{\mathcal{U}}(f)|=\frac{1}{2}.
\end{align*}
Thus, in this case, we have found a function that efficiently distinguishes the outcome from uniform. But note that in some cases such a Pauli string does not exist, such as if the Clifford is $H^{\otimes n}$, as then the outcome distribution is uniform. 

Let us now discuss how to efficiently find the appropriate distinguishing Pauli string. We refer to~\cite{Aaronson_2004} for a review of the basics of the stabilizer formalism.
First, we recall that a stabilizer state $\ket{\psi}$ on $n$ qubits can always be described by $n$ generators $g_1,\ldots,g_n$ of its stabilizer group $SG(\ket{\psi})$. Moreover, note that for a Pauli string
\begin{align*}
\tr{P C\ketbra{0}^{\otimes n}C^\dagger}\in\{-1,1\}.
\end{align*}
is equivalent to $P\in SG(\ket{\psi})$ or $-P\in SG(\ket{\psi})$. Thus, by our previous discussion, the problem of finding a function to distinguish the output of the Clifford circuit from the uniform distribution is equivalent to finding a stabilizer of the state consisting solely of $Z$ and $I$ Pauli operators.

Let us now discuss how to achieve this. First, decompose each generator $g_i$ as the product of a string of Pauli $X$ and Pauli $Z$ matrices plus a global $\pm1$ phase and represent each one of these as vectors $(x_i,z_i,s_i)$ in $\mathbb{F}^{2n+1}_2$. In order to simplify the presentation, we are not going to keep track of the global phase of the elements of the stabilizer group for now. Thus, we restrict to the vectors $(x_i,z_i)$ corresponding to the first $2n$ entries. It is easy to see that if we do not keep track of the global phase, then multiplying two generators is equivalent to adding the corresponding vectors in $\mathbb{F}^{2n}_2$.
Now define the $n\times (2n+1)$ binary matrix $A$ with the vectors $x_i$ in its rows. 
We then have:
\begin{prop}
Let $\ket{\psi}$ be a stabilizer state with generators $g_1,\ldots,g_n$ and corresponding vectors $(x_i,z_i)\in\Z^{2n}$. Then there exists a $Z$ string $P\in SG(\ket{\psi})$ if and only if:
\begin{align}\label{equ:intersectionstab}
\operatorname{span}\{(x_1,z_1),\ldots,(x_n,z_n)\}\cap\lb \{0\}\times \mathbb{F}^{n}_2\rb\not=\{0\}
\end{align}
\end{prop}
\begin{proof}
Note that 
\begin{align*}
\operatorname{span}\{(x_1,z_1),\ldots,(x_n,z_n)\}
\end{align*}
corresponds to the elements of the stabilizer group of the state, up to a global phase. This is because, as discussed before, multiplication in the Pauli group just corresponds to a sum of the vectors, up to the global phase.
Thus, if we find a string of $Z$ in the stabilizer group, then it is also in the intersection in eq.~\eqref{equ:intersectionstab}.
\end{proof}
Thus, we can find the distinguishing Pauli operator by a nonzero element of the subspace
\begin{align*}
\operatorname{span}\{(x_1,z_1),\ldots,(x_n,z_n)\}\cap\lb \{0\}\times \mathbb{F}^{n}_2\rb,
\end{align*}
which can be done by Gaussian elimination. 
\end{document}